\documentclass[12pt,onecolumn]{IEEEtran}
\usepackage{amsmath,amssymb,amsthm,euscript,epsfig,times}
\usepackage{times}

\newcommand{\pr}{{\mathbb{P}}}
\newcommand{\ex}{{\mathbb{E}}}
\newcommand{\E}{{\EuScript{E}}}
\newcommand{\openone}{\leavevmode\hbox{\small1\normalsize\kern-.33em1}}
\newcommand{\Pc}{{\cal{P}}}

\newtheorem{thm}{Theorem}
\newtheorem{lem}{Lemma}
\newtheorem{prop}{Proposition}
\newtheorem*{cor}{Corollary}

\newtheorem{defs}{Definition}

\DeclareMathOperator{\poly}{poly}

\title{Communication under Strong Asynchronism
\thanks{This work was supported in part by NSF under Grant
No.~CCF-0515122, and by a University IR\&D Grant from Draper
Laboratory. }}
\author{Aslan Tchamkerten, Venkat Chandar, and Gregory Wornell\\
Electrical Engineering and Computer
  Science Department \\ Massachusetts Institute of Technology\\
Cambridge, MA 02139, USA \\
Email: $\{$tcham,vchandar,gww$\}$@mit.edu}

\begin{document}

\maketitle

\begin{abstract} 
 We consider asynchronous communication over point-to-point
discrete memoryless channels without feedback. The transmitter starts sending one
block codeword at an instant that is uniformly distributed within a certain time
period, which represents the level of asynchronism.  The receiver, by means of
a sequential decoder, must isolate the message without knowing when the
codeword transmission starts but being cognizant of the asynchronism level. We
are interested in how quickly can the receiver isolate the sent message,
particularly in the regime where the asynchronism level is exponentially larger
than the codeword length, which we refer to as `strong asynchronism.'

 This model of sparse communication might represent the situation of a sensor that
remains idle most of the time and, only occasionally, transmits information to a
remote base station which needs to quickly take action. Because of the limited
amount of energy the sensor possesses, assuming the same cost per transmitted
symbol, it is of interest to consider minimum size codewords given the 
asynchronism level.

The first result is an asymptotic characterization of the largest asynchronism
level, in terms of the codeword length, for which reliable communication can be
achieved: vanishing error probability
can be guaranteed as the codeword length $N$ tends to infinity while the
asynchronism level grows as $e^{N\alpha}$ if and only if $\alpha$ does not
exceed the {\emph{synchronization threshold}}, a constant that admits a simple
closed form expression, and is at least as large as the capacity of the
synchronized channel.

The second result is the characterization of a set of achievable strictly
positive rates in the regime where the asynchronism level is exponential in the
codeword length, and where the rate is defined with respect to the expected
(random) delay between the time information starts being emitted until the time
the receiver makes a decision. Interestingly, this achievability result is
obtained by a coding strategy whose decoder not only operates in an
asynchronously, but has an almost universal decision rule, in the sense
that it is almost independent of the channel statistics.

As an application of the first result we consider antipodal signaling over a
Gaussian additive channel and derive a simple necessary condition between
blocklength, asynchronism level, and SNR for achieving reliable communication. 
\end{abstract}

\keywords
Asynchronous communication, detection and isolation problem, discrete-time
communication, error exponent, low probability of
detection, point-to-point communication, quickest detection, sequential
analysis, sparse communication, stopping
times

\normalsize

\section{Introduction}
\label{intro}

A common assumption in information theory is that `whenever the transmitter
speaks the receiver listens.' In other words, in general, there is the
assumption of perfect synchronization between the transmitter and the receiver
and, basic quantities, such as the channel capacity, are defined under this
hypothesis \cite{Sha2}. In practice this assumption is rarely fulfilled. Time
uncertainty due, for instance, to bursty sources of information often causes
asynchronous communication, i.e., communication for which the receiver has only
a partial knowledge of {\emph{when}} information is sent.

There are, however, notable channels for which asynchronism effects have been
studied from an information theoretic standpoint. An example is the multiple
access channel (see, e.g., \cite{CMP,HuH,Po,Ve}) for which the capacity region
has been computed under various assumptions on the users' asynchronism. Another
important example is the insertion, deletion, and substitution channel for
which only bounds on the capacity are known (see, e.g.,
\cite{AW,D2,DM3,DG}).


In this paper we propose an information theoretic framework that models users'
asynchronism for point-to-point discrete-time communication without feedback.
We consider the situation where the transmitter may start sending information
at a time unknown to the receiver. The time transmission starts is assumed to
be uniformly distributed within a certain interval, which defines the
asynchronism level between the transmitter and the receiver. A suitable notion
of rate is introduced and scaling laws between block message size and
asynchronism level are given for which reliable communication can or cannot be
achieved.\footnote{We refer to `reliable communication' whenever arbitrary low
error probability can be achieved.} Our first result is the characterization of
the highest asynchronism level with respect to the codeword length under which reliable communication can still be
achieved. This limit is attained by a coding strategy that operates at
vanishing rate. This strategy also allows for communication at positive rates
while operating at asynchronism levels that are exponentially larger than the
codeword length. 

In Section \ref{pform} we formally introduce our model and draw connections
with the related `detection and isolation' problem in sequential analysis.
Section \ref{result} contains our main results, Section \ref{analysis} is
devoted to the proofs, and we end with final remarks in Section
\ref{conclusione}. The proofs make often use of large deviations type
bounding techniques for which we refer the reader to \cite[Chapters 1.1 and
1.2]{CK} or \cite[Chapter 12]{CT}.

\section{Problem formulation and background}
\label{pform}

We consider discrete-time communication over a discrete memoryless channel
characterized by its finite input
and output alphabets $\cal{X}$ and $\cal{Y}$, respectively, transition
probability matrix $Q(y|x)$, for all $y\in {\cal{Y}}$ and $x\in {\cal{X}}$, and
`noise' symbol $\star\in {\cal{X}}$ (see
Fig.~\ref{grapheess}).\footnote{Throughout the paper we always assume that for all
$y\in {\cal{Y}}$ there is some $x\in {\cal{X}}$ for which $Q(y|x)>0$.}
\begin{figure}
\begin{center}
\input{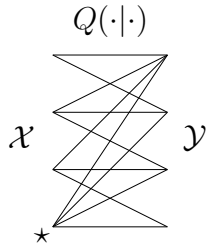}
\caption{\label{grapheess} Communication is carried over a discrete memoryless
channel. When `no information' is sent the input of the channel is the `$\star$' symbol.}
\end{center}
\end{figure}
 The codebook consists of $M\geq 2$ equally likely codewords of length $N$ composed of symbols from
$\cal{X}$ --- possibly also the $\star$ symbol. The transmission of a particular codeword starts at a
random time $\nu$, independent of the codeword to be sent, uniformly distributed in $[1,2,\ldots,A]$,
where the integer $A\geq 1$ characterizes the asynchronism level. We assume that the receiver knows $A$
but not $\nu$. If $A=1$ the channel is said to be synchronized. Throughout the paper, whenever we refer
to the capacity of a channel, it is intended to be the capacity of the synchronized channel. Throughout
the paper we only consider channels $Q$ with strictly positive capacity $C(Q)$.

Before and after the transmission of the information, i.e., before time $\nu$
and after time $\nu+N-1$, the receiver observes noise. Specifically,
conditioned on the value of $\nu$ and on the message to be conveyed $m$, the
receiver observes independent symbols $Y_1,Y_2,\ldots$ distributed as follows.
If $i\leq \nu-1$ or $i\geq \nu+N$, the distribution is   $Q(\cdot|\star)$. At any time $i\in [\nu ,\nu+1,\ldots, 
\nu+N-1]$ the distribution is $Q(\cdot|{c_{i-\nu+1}(m)})$, where $c_{n}(m)$
denotes the $n$\/th symbol of the codeword $c^N(m)$ assigned to message $m$.

The decoder consists of a sequential test $(\tau,\phi)$, where
$\tau$ is a stopping time with respect to the output sequence $Y_1,Y_2,\ldots$\footnote{Recall that a stopping time $\tau$ is an
integer-valued random variable with respect to a sequence of random variables
$\{Y_i\}_{i=1}^\infty$ so that the event $\{\tau=n\}$, conditioned on
$\{Y_i\}_{i=1}^{n}$,  is independent of $\{Y_{i}\}_{i=n+1}^{\infty}$ for all
$n\geq 1$.} indicating when decoding happens, and where $\phi$ denotes a
decision rule\footnote{Formally
$\phi$ is an ${\cal{F}}_\tau$-measurable map where ${\cal{F}}_1,{\cal{F}}_2,\ldots$ is the
natural filtration induced by the process $Y_1,Y_2,\ldots$}
 that declares the decoded message (see Fig.~\ref{grapheesss}).\footnote{In our model
one message
is sent in a certain interval {\emph{with probability
one}}.  An interesting extension of this model that we did not consider is
to give some
probability to the event where no
message is sent. The receiver knows that with some
probability $1-p$ a message starts being sent within a certain interval and that
with probability $p$ no message is sent.}
\begin{figure}
\begin{center}
\input{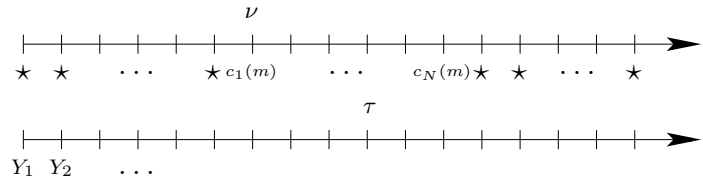}
\caption{\label{grapheesss} Time representation of what is sent (upper arrow)
and what is received (lower arrow). The `$\star$' represents the `noise'
symbol. At time $\nu$ message $m$ starts being sent and decoding occurs at time
$\tau$.}
\end{center}
\end{figure}

We are interested in {\emph{reliable and quick decoding}}. To that aim we first define
the average decoding error probability as
\begin{align*}
\pr({\EuScript{E}})=\frac{1}{A}\frac{1}{M}\sum_{m=1}^M\sum_{l=1}^A \pr_{m,l} ({\EuScript{E}}),
\end{align*}
 where $\EuScript{E}$ indicates the event that the decoded message does not
 correspond to the sent message, and where the subscripts $_{m,l}$ indicate the conditioning on the
event that message $m$ starts being sent at time $l$. Second, we define the
average communication rate with respect to the average delay it takes the
receiver to react to a sent message, i.e.
\begin{align}\label{vilag}
R=\frac{\ln M}{\ex(\tau-\nu)^+}
\end{align} with
  \begin{align*}
\ex(\tau-\nu)^+\triangleq\frac{1}{A}\frac{1}{M}\sum_{m=1}^M\sum_{l=1}^A \ex_{m,l}( \tau-l)^+
\end{align*}
where $x^+$ denotes $\max\{0,x\}$, and where $\ex_{m,l}$ denotes the
expectation with respect to $\pr_{m,l}$.\footnote{Here $\ln $ denotes the natural logarithm.} With the above definitions we now introduce
the notion of achievable rate with respect to a certain asynchronism level as well
as the
notion of {\emph{synchronization threshold}}.
\begin{defs}\label{tasso}
An asynchronism exponent $\alpha$ is {\emph{achievable}} at a rate $R$ if, for any $\varepsilon>0$, there exists a block code
with (sufficiently large) codeword length $N$, operating under asynchronism level $A=e^{(\alpha - \varepsilon)
N}$, while yielding a rate at least as large as $R-\varepsilon$ and an error
probability $\pr({\EuScript{E}})\leq
\varepsilon$. The supremum of the set of asynchronism exponents that are
achievable at rate $R$ is denoted $\alpha(R,Q)$. 
\end{defs}
\noindent Note that, for a given channel $Q$, the asynchronism exponent function
$\alpha(R,Q)$ is
non-increasing in $R$.
\begin{defs}\label{seuil}
The {\emph{synchronization threshold}} of a channel $Q$, denoted by
$\alpha(Q)$, is the supremum of the
set of achievable asynchronism exponents at all rates, i.e.,
$\alpha(Q)=\alpha(R=0,Q)$.
\end{defs}
\noindent Throughout the paper we often use the terminology `coding strategy' or `coding scheme' to denote an infinite sequence of pairs
codebook/decoder labeled by the blocklength. In particular, whenever we
refer to a coding strategy that `achieves a certain rate,' it is intended to be
asymptotically in the limit $N\rightarrow \infty$.

Let us comment on the above bursty communication model and its associated
notions of rate and synchronization threshold.
First observe that we do not introduce a feedback channel from the receiver to the
transmitter. With a noiseless feedback it is possible to inform  the transmitter of the
receiver's decoding time, say in the form of ack/nack, therefore allowing the sending of
multiple messages instead of just one as in our model. Here the noiseless
assumption is crucial. If  the feedback is noisy, the receiver's decision may be wrongly
recognized by the transmitter, which possibly may result in a loss of message
synchronization between transmitter and receiver (say the receiver hasn't yet decoded the
first message while the transmitter has already started to emit the second one). 
Therefore, in order to avoid a potential second source of asynchronism, we omit feedback in our study
and limit transmission to only one message.

The reason for defining the rate with respect to the average delay
$\ex(\tau-\nu)^+$ (see \eqref{vilag}) is motivated by the following considerations. At first sight,
a natural measure of delay may be the codeword length $N$. However, in light of
the use of sequential decoding, the codeword
length does not provide a measure of the delay needed for the information to be
reliably decoded. Another candidate for the delay one might consider is
$\ex(\tau)$ or, equivalently, $\ex \nu+\ex(\tau -\nu)$. The fact that this
delay takes into account the initial offset $\ex \nu$ can be regarded as a
weakness since this offset can be influenced neither by the transmitter nor by
the receiver. Also, with such a delay measure, in the regime of positive
asynchronism exponents we are interested in, the rate is always
(asymptotically) vanishing for any reliable coding strategy.\footnote{To see this consider the rate
defined as $\ln M/(\ex \nu+\ex(\tau -\nu))$. To achieve vanishing error
probability as $M$ (or $N$) tends to infinity, the reaction delay $\ex(\tau
-\nu)$ must grow at least linearly with $\ln M$ (if not this would imply that
reliable communication above capacity would be possible). Similarly, $M$ and
$N$ must satisfy $N\geq \ln M$. Also, in the regime of positive asynchronism
exponents, i.e., when $A=e^{N\alpha}$ for some $\alpha>0$, we have $\ex
\nu=e^{N\alpha}/2$ since $\nu$ is uniformly distributed in $[1,2,\ldots,A]$. 
Therefore, in the regime of positive asynchronism exponents, the rate $\ln
M/(\ex \nu+\ex(\tau -\nu))$ is vanishing as $N\rightarrow \infty$ for any
coding strategy that achieves arbitrarily low error probability.}  Instead, we
propose to consider $\ex(\tau-\nu)^+$, the average time the transmitter needs
to wait until the receiver makes a decision. 
Also note that, in the definition of achievable rate (Definition \ref{tasso}), we choose to grow $A$ with $N$. Indeed, when $A$ is fixed
the problem becomes trivial. By using sufficiently long codewords and simply decoding at the
(fixed) time $A+N-1$ the asynchronism effect on the rate can be made negligible.

We now briefly discuss the notion of synchronization threshold. This 
threshold is defined with respect to zero rate coding strategies, that is
strategies for which $\ln M/\ex(\tau-\nu)^+$ tends to zero (as $N\rightarrow
\infty$). However, because $\ex(\tau-\nu)^+$ and $N$ need not coincide in
general, zero rate coding strategies need not, in general, yield a vanishing
fraction $\ln M/N$ as $N$ tends to infinity. Indeed, as we will see, one can
operate arbitrarily closely to the synchronization threshold while having $\ln
M/N$ asymptotically bounded away from zero.



Perhaps the closest sequential decision problem our model relates to is a
generalization of the change-point problem, often called the `detection and
isolation problem,' introduced by Nikiforov in $1995$ (see \cite{N,Lai2} and
\cite{BN} for a survey). A process $Y_1,Y_2,\ldots$ starts with some initial
distribution and changes it at some unknown time. The post change
distribution can be any of a given set of $M$ distributions. By sequentially
observing $Y_1,Y_2,\ldots$ the goal is to quickly react to the statistical
change and isolate its cause, i.e., the post-change distribution. Hence, our
synchronization problem takes the form of a detection and isolation problem
where the change in distribution is induced by the transmitted message. However, to the
best of our knowledge studies related to the detection and isolation problem
usually assume that once the observed process jumps into one of its post-change
distributions, it remains in that state forever. This means that, eventually,
if we wait long enough, a correct decision is be possible. Instead, in the
synchronization problem the change in distribution is {\emph{local}} since it
only lasts the duration of a codeword length. In particular once the codeword
is `missed' no recovery is possible. Finally, optimal decoding rules for the
detection and isolation problem seem to have been obtained only in the limit of
small error probabilities $\pr(\EuScript{E})$ while keeping $M$, the number of post-change
distributions, fixed.\footnote{Here optimal decoding rules refer to sequential
tests yielding minimum reaction delay, usually a
function of $\tau-\nu$, given a certain error probability.} In our case we 
typically let $M$ grow as $(1/\pr(\EuScript{E}))^\xi$, for some $\xi>0$.

\section{Results}
\label{result}
 Our first result is the characterization of the synchronization threshold.
\begin{thm}\label{unow} For any discrete memoryless channel $Q$, the
synchronization threshold as given in Definition \ref{seuil} is given by
$$\alpha(Q)=\max_x
D(Q(\cdot|x)||Q(\cdot|\star))$$
where $D(Q(\cdot|x)||Q(\cdot|\star))$ is the divergence (Kullback-Leibler
distance) between
$Q(\cdot|x)$ and $Q(\cdot|\star)$. Furthermore, any synchronization threshold
$\alpha<\alpha(Q)$ can be achieved by a coding strategy that yields
$\lim_{N\rightarrow\infty}\ln
M/N>0$.
 \end{thm}
\noindent The theorem says that vanishing error probability can be achieved as
the blocklength $N$ tends to infinity if the
asynchronism level grows as $e^{N\alpha}$ where $\alpha<
D(Q(\cdot|x)||Q(\cdot|\star))$. Conversely, any coding strategy that operates at an
asynchronism exponent  $\alpha>
D(Q(\cdot|x)||Q(\cdot|\star))$ cannot achieve arbitrary low error probability. 
The second part of the theorem shows the distinction between the delay measured by the codeword length $N$ and by the expected 
`reaction time' $\ex(\tau-\nu)^+$. Arbitrary closely to the synchronization
threshold one can (asymptotically) guarantee $\ln M/N$ to be strictly positive,
while the question remains open for the rate $\ln M/\ex(\tau-\nu)^+$.
Specifically, it
remains to be seen whether $\alpha(Q)=\lim_{R\downarrow 0}\alpha(R,Q)$ (assuming
$\alpha(Q)<\infty$). This issue will be discussed in Section \ref{rate0}.

At least some connections between channel capacity and synchronization threshold exist. Although these two quantities are not directly
related, both refer to limits on hypothesis discrimination. The first concerns a
purely isolation problem whereas the second concerns an almost purely
detection problem (since there is no rate constraint). 
It may be interesting to note that the
synchronization threshold $\alpha(Q)$ is always at least as large as $C(Q)$.
To see this let $P$ be the capacity achieving distribution of the (synchronized) channel $Q$. It is well known \cite[Lemma 13.8.1]{CT} that for any distribution $V$ on $\cal{Y}$
$$D(PQ||PP_Y)\leq D(PQ||PV) $$
where $P_Y$ is the right marginal of $PQ=P(\cdot)Q(\cdot|\cdot)$.
Letting $V=Q(\cdot|\star)$ we get
\begin{align*}
C(Q)&\triangleq D(PQ(\cdot|\cdot)||PP_Y)\\
&\leq  D(PQ(\cdot|\cdot)||PQ(\cdot|\star))\\
&=\sum_x P(x)\sum_yQ(y|x)\ln \frac{Q(y|x)}{Q(y|\star)}\\
&\leq \max_x D(Q(\cdot|x)||Q(\cdot|\star))\\
&=\alpha(Q)
\end{align*}
Finally it can be checked that if $C(Q)=0$ then $\alpha(Q)=0$.

\subsection*{Example: the Gaussian channel} \noindent As an application of
Theorem \ref{unow} we consider antipodal signaling over a Gaussian channel and
derive a necessary condition between asynchronism level, block length, and
signal to noise ratio (SNR) for achieving reliable communication. Suppose
communication takes place over an additive channel $X\rightarrow Y=X+Z$ where
$X$ denotes the input, $Y$ the output, and where $Z$ is a normally distributed
random variable,  independent of $X$, with zero mean and unit variance.
We consider antipodal signaling, that is $c_i(m)=\pm \sqrt{\text{SNR}}$ for
all $i\in \{1,2,\ldots,N\}$ and $m\in \{1,\ldots,M\}$, where the SNR is some
positive constant. Before decoding, the receiver makes a hard decision on each
received symbol and declares $+1$ if $Y_i\geq 0$ and $-1$ if $Y_i<0$. The noise
symbol $\star$ equals zero meaning that when no information is sent the
receiver declares $+1$ or $-1$ with probability $1/2$.  The inputs
$+\sqrt{\text{SNR}}$ and $-\sqrt{\text{SNR}}$ are received correctly with
probability $1-\varepsilon$  and are flipped with probability $\varepsilon$,
where $\displaystyle \varepsilon=e^{-\frac{\text{SNR}}{2}(1+o(1))}$ as the SNR
tends to infinity. The discrete channel $Q$ that results from the hard
decision procedure is depicted in Fig.~\ref{gaussian}.  \begin{figure}
\begin{center} \input{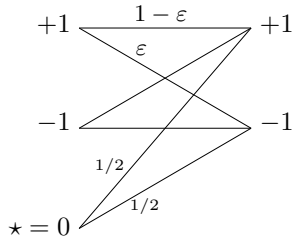} \caption{\label{gaussian} Antipodal signaling over a
 Gaussian channel with hard decision
at the decoder.} \end{center} \end{figure} From Theorem \ref{unow}, any coding
strategy that yields vanishing error probability satisfies
$\limsup_{N\rightarrow \infty}1/N \ln A \leq \alpha(Q)$ where  \begin{align*}
\alpha(Q)&=\max_x D(Q(\cdot|x)||Q(\cdot|\star))\\ &=\ln 2 - H(\varepsilon)\\
&=\ln 2 -H(e^{-\frac{\text{SNR}}{2}(1+o(1))}) \quad \text{as SNR
}\rightarrow\infty \end{align*} with $H(\varepsilon)\triangleq -\varepsilon \ln
\varepsilon -(1-\varepsilon)\ln (1-\varepsilon)$. Therefore, as $N$ tends to
infinity, in order to achieve reliable communication it is necessary that
$$\frac{1}{N} \ln A \leq \ln 2 -H(e^{-\frac{\text{SNR}}{2}(1+o_1(1))})+o_2(1)$$
where $o_1(1)$ and $o_2(1)$ are vanishing functions of the SNR and of $N$,
respectively. Because of the chosen quantization, in the limit of high SNR we
have $\frac{1}{N} \ln A \overset{\sim}{\leq} \ln 2$, and an increase in the
power results in a negligible increase of the asynchronism level for which
reliable communication is possible  (for fixed blocklength). To exploit power
at high SNR it is necessary to have a finer quantization at the output. Finally
notice that for this (quantized) channel the synchronization threshold
coincides with the channel capacity.\hfill{$\square$}

While we do not characterize the asynchronism exponent function $\alpha(R,Q)$ for $R>0$, Theorem \ref{ra}
provides a non trivial lower bound characterization of $\alpha(R,Q)$, for any
$R\in [0,C(Q))$.

We use the
notation $(PQ)_Y$ to denote the right marginal of a joint distribution
$P(\cdot)Q(\cdot|\cdot)$ and, given a joint distribution $J$ on
${\cal{X}}\times {\cal{Y}}$ we denote by $I(J)$ the mutual information induced
by $J$. Also we denote by ${\cal{P}}^{{\cal{Y}}|{\cal{X}}}$ the set of
conditional distributions of the form $V(y|x)$ with $x\in {\cal{X}}$ and $y\in {\cal{Y}}$.

\begin{thm}\label{ra}
Let $Q$ be a discrete memoryless channel. If for some constants $\alpha\geq 0$,
$t_1\geq 0$, $t_2>1$, and input distribution $P$, with $I(PQ)>0$, the following inequalities
\begin{align*}
a.\quad &\alpha <\inf_{\substack{V\in {\cal{P}}^{{\cal{Y}}|{\cal{X}}}\\
D((PV)_Y||Q(\cdot|\star))<
\frac{t_1\alpha}{\delta(t_1+t_2-1)}}}D((PV)_Y||(PQ)_Y)\nonumber \\
b.\quad &\alpha <\min_{\substack{V
\in {\cal{P}}^{{\cal{Y}}|{\cal{X}}}\\ I(PV)\leq \frac{t_2\alpha}{\delta(t_1+t_2-1)}}} D(PV||PQ)\nonumber \\
c.\quad &\frac{t_1}{t_2}<
\frac{D((PQ)_Y||Q(\cdot|\star))}{I(PQ)}
\end{align*}
are satisfied for some $\delta \in (0,1)$, then the rate $I(PQ)/t_2$ is achievable
at an asynchronism exponent $\alpha$.
\end{thm}
\noindent Note that the conditions $a$ and $b$ in Theorem \ref{ra} are easy to check
numerically since they only involve convex optimizations. 
Also notice, on the right hand side of the inequality $b$, the sphere packing exponent
function --- of the channel $Q$ with input distribution $P$ --- 
evaluated at $\frac{t_2\alpha}{\delta(t_1+t_2-1)}$ (see \cite[p.166]{CK}).
\begin{cor}
For any channel $Q$ with capacity
$C(Q)>0$, any rate $R\in (0,C(Q))$ can be achieved at a strictly positive
asynchronism exponent.
\end{cor}
\begin{proof}[Proof of the Corollary]
Consider the inequalities $a$, $b$, and $c$ from Theorem \ref{ra}. First
choose some $P$ and $t_2>1$ so that $I(PQ)/t_2\geq R$ and so that $(PQ)_Y\ne Q(\cdot|\star)$
(this is always possible since $C(Q)>0$). By setting $t_1=0$ the inequality $c$ holds (since
its right hand side is strictly positive). Also inequality $a$ holds for any finite $\alpha$
(the infimum equals infinity). For the inequality $b$, observe that its right hand side is a
decreasing function of $\alpha$ and has a strictly positive value at $\alpha =0$ (since
$I(PQ)>0$). It follows that inequality $b$ holds for strictly positive and small enough values
of $\alpha$.  
\end{proof}

\subsection{Coding for asynchronous channels}\label{achievability}
In this section we present the coding scheme from which one deduces
Theorem \ref{ra} and the direct part of Theorem \ref{unow}.
As we will see, our scheme does not subdivide the
synchronization problem into a detection problem followed by a message
isolation problem: detection and isolation are treated jointly.

The codebook is randomly generated according to some distribution $P$. If the
aim is only to reliably communicate at a certain asynchronism exponent
$\alpha$, there is some degrees of freedom in choosing $P$. One possible
choice is to pick a $P$ that satisfies $$D((PQ)_Y||Q(\cdot|\star))+I(PQ)-\ln
M/N > \alpha$$ with $D((PQ)_Y||Q(\cdot|\star))>0$ and $I(PQ)>0$, where $M$ represents the size of the message set
and $N$ the size of the codewords (see proof of Proposition \ref{exca}).  In the regime
where the asynchronism exponent is close to $\alpha(Q)$ the codewords are
mainly composed of the symbol  $\arg\max_x D(Q(\cdot|x)||Q(\cdot|\star))$. 
Indeed, in this asynchronism regime, the main source of error comes from a miss
detection of the sent codeword, later referred to as `false-alarm.' We deal with this source of error by
distillating information using codewords with (mostly) symbols that induce
output distributions that are `as far as possible' from the output distribution
induced by the $\star$ symbol. Finally if the aim is to accommodate both rate and
asynchronism constraints, the distribution $P$ has to satisfy the conditions
explicitly stated in Theorem \ref{ra}.

For the decoder, let us observe first that our communication model admits two sources of error. The first
comes from an atypical behavior of the noise during the period when no information is conveyed, which may
result in a false-alarm. The second comes from an atypical behavior of the channel during information
transmission, which may result in a miss-isolation of the sent codeword. These two sources of error
depend on the asynchronism level as well as on the communication rate: the higher the asynchronism the
higher the first source of error, the higher the communication rate the higher the second source of
error. Accordingly, our decoder is the combination of two criteria parameterized by constants that are
chosen based on the level of asynchronism and according to the rate we aim at.

More specifically,  the decoder observes the channel outputs $Y_1,Y_2,\ldots$  and makes a decision
as soon as it observes $i$ consecutive output symbols, with $i\in [1,2,\ldots,N]$,
that simultaneously satisfy two conditions. The first condition is that these
symbols should look `sufficiently different' from the noise, as measured by
the divergence. The second condition is that these symbols
must be sufficiently correlated, in a mutual information sense, with one of the
codewords. We formalize this below.

For $j\geq i$ we write $x_i^j$ for $x_i,x_{i+1},\ldots,x_j$. If $i=1$ we use the shorthand
notation $x^j$ instead of $x_i^j$. Given a pair $(x^n,y^n)$ let us denote by $\hat{P}_{(x^n,y^n)}$ the
empirical distribution of $(x^n,y^n)$, i.e., $\hat{P}_{(x^n,y^n)}(x,y)=\frac{1}{n}\sum_{i=1}^n
\openone_{(x,y)}(x_i,y_i)$ where $\openone_{(x,y)}(x_i,y_i)=1$ if $(x_i,y_i)=(x,y)$, else equals zero. To
each message $m\in [1,2,\ldots,M]$ associate the stopping time\footnote{\label{fosnote}It may seem to the
reader that the mutual information condition in \eqref{tdecod} given by \begin{align}\label{fsg5}
\min_{k\in [1,\ldots,i]} \left[k I(\hat{P}_{c^k(m),y_{n-i+1}^{n-i+k}}) +(i-k)
I(\hat{P}_{c_{k+1}^i(m),y_{n-i+k+1}^n})\right] \geq t_2 \ln M \Bigg\} \end{align}is convoluted, and that
it could be replaced, for instance, by \begin{align}\label{fsg6} i I(\hat{P}_{c^i(m),y_{n-i+1}^{n}}) \geq
t_2 \ln M \;. \end{align} Our choice is motivated by a technical consideration related to the
false-alarm event induced by $i$ last symbols that are generated partly inside and partly outside the transmission
period (see Case II of the proof of Lemma \ref{false-alarm}). }

\begin{align}\label{tdecod}
&\tau_m=\inf\Bigg\{n\geq 1\;:\; \exists  i\in \{1,\ldots,N\} \:\: \text{so that}
\;iD(\hat{P}_{Y_{n-i+1}^n}||Q(\cdot|\star))\geq t_1 \ln M \:\text{and}\:\nonumber \\
&\min_{k\in [1,\ldots,i]} \left[k I(\hat{P}_{c^k(m),y_{n-i+1}^{n-i+k}})
+(i-k) I(\hat{P}_{c_{k+1}^i(m),y_{n-i+k+1}^n})\right]
\geq t_2 \ln M \Bigg\}
\end{align}
where $t_1\geq 0$ and $t_2>1$ are some fixed threshold constants to be appropriately
chosen according to the asynchronism level and desired communication rate.
The decoding is made at time $$\tau=\min_{m\in [1,2,\ldots,M]}\tau_m$$ and the
message $\bar{m}$ that is declared is any that satisfies
$\tau_{\bar{m}}=\tau$.

It should be emphasized that there may be other sequential decoders that also
achieve the synchronization threshold. The one we propose has the property that
it also allows for communication at positive rates and positive asynchronism
exponents. Also, an interesting feature of the above decoder is that, in
addition to operating in an asynchronous setting, it is also almost universal
in the sense that its rule does not depend of the channel statistics, except
for the noise distribution $Q(\cdot|\star)$. In fact this decoder is an
extension of a sequential universal decoder introduced in \cite[eq.~(10)]{TT2}
for the synchronized setting.

In the context of asynchronous communication, the same decoding
rule as above is considered in \cite{TKW}, but without the
divergence condition, i.e., a decision is made as soon as for some $m$ and $i$ the condition
$$\min_{k\in [1,\ldots,i]}
\left[k I(\hat{P}_{c^k(m),y_{n-i+1}^{n-i+k}}) +(i-k)
I(\hat{P}_{c_{k+1}^i(m),y_{n-i+k+1}^n})\right] \geq t_2 \ln M \Bigg\}$$
holds. With the mutual information condition alone, however, it
was not possible to prove that reliable
communication can be achieved for asynchronism exponents higher than the capacity of the channel.

\subsection{Continuity of $\alpha(\cdot,Q)$ at $R=0$}\label{rate0}
We discuss the continuity of $\alpha(\cdot,Q)$ at $R=0$ in light of Theorem
\ref{ra}. The right hand side of inequality $b$, the sphere packing bound, is
associated to the miss-isolation error event of the sent codeword
associated with the coding scheme discussed in \ref{achievability} (this will be
seen in the proof of Theorem \ref{ra}). Therefore, regardless of the rate, any achievable synchronization 
exponent $\alpha$ obtained via Theorem \ref{ra} is bounded by the sphere packing exponent
at zero rate, which can
be smaller than the synchronization threshold (see Fig.~\ref{bscstar} for an
example). This motivates the
conjecture that $\alpha(Q,0)\ne\lim_{R\downarrow 0}\alpha (Q,R)$ in general.
\begin{figure}
\begin{center} \input{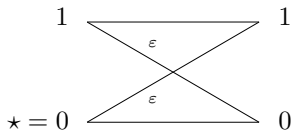} \caption{\label{bscstar} A binary symmetric channel
has a sphere packing bound at zero rate,
$E_{sp}(R=0,Q)$ given by $\max_P\min_{V:I(PV)=0}D(PV||PQ)$, that
can be smaller compared to $\alpha(Q)$. Specifically, Theorem \ref{unow} yields
$\alpha(Q)=\varepsilon \ln [\varepsilon/(1-\varepsilon)]+ (1-\varepsilon)\ln
[(1-\varepsilon)/\varepsilon]$ and it can be checked that
$E_{sp}(R=0,Q)\leq 0.5 \ln [0.5/(1-\varepsilon)]+0.5 \ln [0.5/\varepsilon]$.
Therefore $E_{sp}(R=0,Q)\leq 0.5 (1+o(1))\alpha(Q)$ as $\varepsilon \rightarrow 0$.} \end{center} 
\end{figure}

Note that there are channels for which the asynchronism exponent function is continuous at zero
rate, such as the one given in Fig.~\ref{eeszi}. 
\begin{figure}
\begin{center} \input{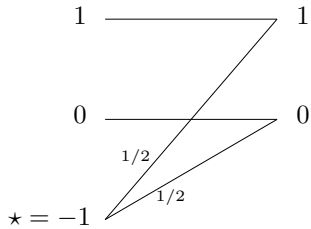} \caption{\label{eeszi} Example of a channel for
which $\alpha(Q,0)=\lim_{R\downarrow 0}\alpha (Q,R)$.} \end{center} 
\end{figure}
Indeed, in this case $\alpha(Q)=\ln 2$ by Theorem \ref{unow}. Then,
considering the three inequalities given in Theorem \ref{ra}, let $t_1=0$ and let the input distribution $P$ be
defined as $P(1)=p=1-P(0)$ for some fixed $p\in (0,1/2)$. With this choice of $t_1$
and $P$ the inequality $a$ holds for any finite $\alpha$ (the infimum is infinite) and inequality $c$ holds
for any $t_2>1$ since its right hand side is strictly positive. We now focus on the inequality $b$. Observe
that any channel $V\ne Q$ with inputs $0$ and $1$ gives $D(PV||PQ)=+\infty$.  Therefore, for any $\delta\in
(0,1)$ and $t_2>1$ the right hand side of the inequality $b$ is infinite if $Q$ satisfies 
\begin{align}\label{yap1}
\frac{t_2\alpha}{\delta(t_2-1)}<I(PQ)\;,
\end{align}
and zero otherwise. Now pick an arbitrarily small $\mu>0$ and choose $P$ with $p$ sufficiently close to $1/2$ so that 
\begin{align}\label{yap2}
I(PQ)\geq \alpha (Q)-\mu/2\;.
\end{align}
We conclude from \eqref{yap1} and \eqref{yap2} that, by choosing $\delta$ close
enough to one and $t_2$ large enough,
any asynchronism exponent
$$\alpha\leq \alpha(Q)-\mu$$
can be achieved at all rates up to $I(PQ)/t_2$.

\section{Analysis}\label{analysis}
 In this section we prove the converse and
the direct part of Theorem \ref{unow}. The converse shows that no coding
strategy achieves vanishing error probability while operating at an
asynchronism exponent higher than $\alpha(Q)$. For the direct part we show that
the coding scheme proposed in Section \ref{achievability} can reliably operate
arbitrarily closely to the asynchronism exponent $\alpha(Q)$. By extending the
analysis of this scheme we will prove Theorem \ref{ra}. The difference between
the achievability schemes of Theorem \ref{unow} and \ref{ra} lies in the
codebooks. For Theorem \ref{unow} the codebook is randomly generated according
to a certain distribution $P$, while for Theorem \ref{ra} we impose that each codeword is (essentially) of constant composition
$P$ uniformly over its length. 

\begin{prop}[Converse]\label{converse} Suppose that $Q(y|\star)>0$ for all $y
\in {\cal{Y}}$. Then no coding strategy achieves an asynchronism exponent strictly
greater than $$\max_{x \in \cal{X}} D(Q(\cdot|x)||Q(\cdot|\star))\;.$$
\end{prop}

\noindent Proposition \ref{converse} assumes that $Q(y|\star)>0$
for all $y \in {\cal{Y}}$. Indeed, if $Q({y}|\star)=0$ for some
${y}\in \cal{Y}$ it will shown in Proposition \ref{exca} that reliable communication can be
achieved irrespectively of the
exponential growth rate of the asynchronism level with respect to the
blocklength.

\begin{proof}[Proof of Proposition \ref{converse}]
Suppose there are two equally likely messages, $m$ and $m'$, and that the
decoder is given the sequence of maximal length $y_1,y_2,\ldots,y_{A+N-1}$. We make the
hypothesis that each codeword
$c(m)$ and $c(m')$ uses one symbol repeated $N$ times. The case where each codeword
uses
multiple symbols is obtained by a straightforward extension of the single symbol
case and is therefore omitted. 
Also, we optimistically assume that the receiver is cognizant of the fact
that the sent message is delivered during one of the $s$ distinct time slots of duration $N$,
where $s$ is the integer part of $(A+N-1)/N$, as shown in Fig.~\ref{graphees2}.
\begin{figure}
\begin{center}
\input{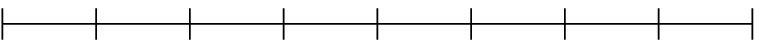}
\caption{\label{graphees2} Parsing of the received sequence of maximal length $A+N-1$
into $s$ blocks $y^{(1)},y^{(2)},\ldots,y^{(s)}$ of
length $N$, where $s$ is the integer part of $(A+N-1)/N$.}
\end{center}
\end{figure}
An easy computation shows that, given a sequence $y^{A+N-1}$, the maximum a posteriori
decoder declares message $m$ or $m'$ depending whether the sum
$$\sum_{l=1}^s z(y^{(l)})$$
is positive of negative,\footnote{If the sum is zero the decoder declares one of the two messages at
random.} with
\begin{align}\label{zhete}
z(y^{(l)})\triangleq \left[\frac{Q(y^{(l)}|c(m))}{Q(y^{(l)}|\star)}-\frac{Q(y^{(l)}|c(m'))}{Q(y^{(l)}|\star)}\right]
\end{align}
and where $Q(y^{(l)}|c(m))$ denotes the probability of the $l$\/th block
$y^{(l)}$ of size $N$ given the codeword $c(m)$, and where $Q(y^{(l)}|\star)$ refers to the
same probability now conditioned on the string of $N$ consecutive $\star$. The probability of the error event $\EuScript{E}$ is hence lower
bounded as
$$\pr({\EuScript{E}})\geq \frac{1}{2}\left[ \pr_m\left(\sum_{l=1}^s
z(Y^{(l)})<0\right)+\pr_{m'}\left(\sum_{l=1}^s z(Y^{(l)})>0\right)\right]$$ where
$\pr_m$ refers to the probability conditioned on message $m$ being
sent. Note that under $\pr_m$ and $\pr_{m'}$ the $z(Y^{(l)})$ are all i.i.d. according to the noise
distribution except for $z(Y^{(\nu)})$ whose distribution depends on the
sent message. 

 Let $T_{m}$ be the set of sequences $y^{N}$ that are
strongly typical with respect to $Q(\cdot|c(m))$ \cite[p.33]{CK}, i.e, any
sequence $y^N\in T_{m}$ satisfies $|n(y;y^N)/N-Q(y|c(m))|<\mu$ where
$n(y;y^N)$ is the number of times the symbol $y$ appears in $y^N$. We choose
the strong typicality constant $\mu$ to be so that $0<\mu\ll 1$ and the
blocklength $N$ large enough that $\pr_m(Y^{(\nu)}\in T_m)\geq 1-\mu$. We define
$T_{m'}$ analogously. Further, we define $h$ to be equal to $\max_{y^N\in T_m\cup
T_{m'}}|z(y^N)|$. Using the independence of $z(Y^{(\nu)})$ and
$\sum_{l\ne \nu}z(Y^{(l)})$ under $\pr_m$ we get
\begin{align*}
\pr_m\left(\sum_{l=1}^s z(Y^{(l)})<0\right)&\geq \pr_m\left(\{Y^{(\nu)}\in
T_m \}\cap \left\{ \sum_{l\ne \nu}z(Y^{(l)})<-h
\right\}\right) \\
&\geq (1-\mu) \pr \left( \sum_{l=1}^{s-1}z(Y^{(l)})<-h \right)\;.
\end{align*}
The sum in the argument of the last term above involves $s-1$
independent random variables distributed according to $Q(\cdot|\star)$. 
For simplicity from now on we denote these random variables by $Z_l$
instead of $z(Y^{(l)})$. We then deduce that
\begin{align}\label{quart}
\pr({\EuScript{E}})\geq \left(\frac{1-\mu}{2}\right)\pr\left( \big|\sum_{l=1
}^{s-1}Z_l\big|> h \right) \;.
\end{align}
In the remaining part of the proof we show that, if
$s=e^{(\alpha(Q)+\varepsilon)N }$, with $\varepsilon>0$, the random walk
$\sum_{i=1}^{s-1}Z_l$ crosses $h$ with finite probability as $N$ tends to
infinity, proving the Proposition. 
At the core of the argument lies the following Lemma whose proof is deferred to the Appendix.
\begin{lem}\label{bucket}
Let $P$ be a distribution over some finite alphabet
${\cal{A}}=\{a_1,a_2,\ldots,a_{|{\cal{A}}|}\}$ and suppose that for some
integer $s\geq 1$ $$\frac{3}{s\delta_0}<\min\{P(a_1),P(a_2)\}$$ for some constant
$\delta_0\in (0,1)$. Let ${\hat{P}}$ be an empirical type\footnote{An empirical
type over ${\cal{A}}^s$ is a distribution $\hat{P}$ over $\cal{A}$ so that
$\hat{P}(a)$ is an
integer multiple of $1/s$, for all $a\in
\cal{A}$.} over
${\cal{A}}^s$ so that
$\min\{\frac{{\hat{P}}(a_1)}{{P}(a_1)},\frac{P(a_2)}{{\hat{P}}(a_2)}\}\geq
\delta_0$ and ${\hat{P}}(a_2)\geq 1/s$.  Let $\bar{P}$ be defined
so that $\bar{P}(a_1)=\hat{{P}}(a_1)-\frac{3}{s}$, $
\bar{P}(a_2)={\hat{P}}(a_2)+\frac{3}{s}$, and $
\bar{P}(a_i)={{\hat{P}}}(a_i)$ for any $a_i\in
{\cal{A}}\backslash\{a_1,a_2\}$. Then $$P^s(T(\bar{P}))\geq
\delta P^s(T(\hat{{P}}))$$ for some strictly positive constant
$\delta=\delta(\delta_0)$, where $P^s$ denotes the product
distribution induced by $P$ over ${\cal{A}}^s$, and where
$T({\hat{P}})$ and $T(\bar{P})$ denote the set of sequences of
length $s$ with empirical type ${\hat{P}}$ and $\bar{P}$,
respectively.
\end{lem}
\noindent We use the lemma
with ${\cal{A}}=\{a: a=z(y^N)\:\text{for some}\: y^N\in {\cal{Y}}^N\}$, $s$ defined
as the integer part of $e^{N(\alpha+\varepsilon)}$ for some arbitrary
$\varepsilon>0$, and
$P$ defined as $P(a)=\sum_{y^N:z(y^N)=a}Q(y^N|\star)$ for all $a\in
 {\cal{A}}$. Also, we let $a_1=h$, $a_2$ be the symbol in $\cal{A}$ with the
 highest probability under $P$, and ${\hat{P}}$ be any distribution on
$\cal{A}$ so that $\big|1-\frac{{\hat{P}}(a_i)}{P(a_i)}\big|<\mu$ for
$i\in\{1,2\}$. In
the sequel we label such distributions ${\hat{P}}$ as `typical types.'
  We  now assume that $s$, $P$, ${\hat{P}}$, $a_1$, and $a_2$ satisfy the hypothesis of Lemma \ref{bucket} and
 will show it at the end of the proof. 
 
Suppose by contradiction that the right hand side of \eqref{quart} goes to
zero as $N\rightarrow \infty$, i.e., that 
\begin{align}\label{quart5}
P^s\left( \big|\sum_{l=1
}^{s}Z_l\big|\leq  h \right) \geq 1-\rho
\end{align}
for any arbitrary $\rho>0$ and $N$ large enough. Assume for the moment that \eqref{quart5} implies for
$N$ large enough
\begin{align}\label{quart3}
P^{s}\left(\bigg\{\big|\sum_{l=1
}^{s}Z_l\big|\leq h\bigg\}\cap \bigg\{ Z^s\:\: \text{has a typical type
}\hat{P}\bigg\}\right)\geq 1-\mu-\rho\;.
\end{align}
This implication will be shown at the end of the proof. 
Now, for a given typical
type ${\hat{P}}$ let $\bar{P}$ be defined as in Lemma
\ref{bucket}. Observe that if $Z^s$ belongs to the event 
$$\bigg\{\big|\sum_{l=1
}^{s}Z_l\big|\leq h\bigg\}\cap \bigg\{ Z^s\:\: \text{has typical type }\hat{P}\bigg\}$$
then $Z^s$ has a type ${\hat{P}}$ that yields a $\bar{P}$ whose type
class\footnote{The type class of $\bar{P}$ is the set of all sequences $z^s$
that have type $\bar{P}$.} belongs to the event\footnote{This step follows by
noting first that $a_1 = 
e^{N D((Q(\cdot | x) || Q(\cdot | \star)) (1 + o(1))}$ as $\mu\rightarrow 0$ and
$N\rightarrow \infty$, and second that $a_2/a_1=o(1)$ as $N\rightarrow \infty$
(for $\mu>0$ small enough).}
$$\bigg\{\big|\sum_{l=1
}^{s}Z_l\big|>h\bigg\}\;.$$
Hence, from Lemma \ref{bucket} and \eqref{quart3} there exists some $\delta>0$
so that \begin{align}\label{quart4}
P^{s}\left(\big|\sum_{l=1
}^{s}Z_l\big|> h\right)\geq \delta(1-\mu-\rho)
\end{align}
for $N$ large enough,
which is in contradiction with \eqref{quart5} for $\rho$ small enough.
We conclude that $\pr\left( \big|\sum_{l=1
}^{s}Z_l\big|>  h \right)$ is asymptotically bounded away from zero,
and so is the right hand side of \eqref{quart}.

To conclude the proof we need to justify the steps from \eqref{quart5} to
\eqref{quart3} and we need to check that $P$ and ${\hat{P}}$ satisfy the hypothesis of the lemma with our choice
of $a_1$ and $a_2$. For this last check, first note
that $z(y^N)$ depends only on the type of $y^N$. Without loss of generality we
assume that $h$ is achieved by a type in $T_m$. Hence we have\footnote{Throughout the paper
we use the notation
$\poly(N)$ to denote any term that is either a polynomial in $N$ or the
 inverse of a polynomial in $N$.}
\begin{align*}
P(a_1)&\triangleq\sum_{y^N\in T_m}Q(y^N|\star)\\
&\geq
e^{-ND(Q(\cdot|x)||Q(\cdot|\star))(1+\eta)}\poly (N)
\end{align*}
where $x$ is the $N$ times repeated symbol for the codeword $c(m)$, and where
$\eta=\eta(\mu)>0$ goes to zero as $\mu$ vanishes.
It follows that $s P(a_1)$ grows
exponentially with $N$ provided $\mu$ is small enough. Thus the condition
$1/(s \delta_0) <P(a_1)$ is trivially satisfied for any $\delta_0\in (0,1)$.
Also, our choice of $a_2$ gives $1/(s \delta_0) <P(a_2)$
for any $\delta_0$. This is because $P(a_2)\geq \poly (N)$ since there are
polynomially many types of length $N$ and that $a_2$ is generated by the type of
the highest probability. Finally, that the conditions $\min\{\frac{{\hat{P}}(a_1)}{{P}(a_1)},\frac{P(a_2)}{{\hat{P}}(a_2)}\}\geq
\delta_0$ and ${\hat{P}}(a_2)\geq 1/s$ are satisfied follows from the definition
of ${\hat{P}}$. 

Finally we show that
$$P^{s}\left( Z^s\:\: \text{has typical type }\hat{P}\right)\;$$
can be made arbitrarily close to one as $N$ tends to infinity, justifying the step from
\eqref{quart5} to \eqref{quart3}. Using Chebyshev's inequality and the fact that the variance
of a binomial is dominated by its expectation we get\footnote{Here
$\openone_{a_1}(Z_l)$ equals 1 if $Z_l=a_1$, zero else.}
\begin{align*}
P^{s}\left( \bigg| \frac{\hat{P}_{Z^s}(a_1)}{P(a_1)}-1\bigg| \geq \mu\right)&=
P^{s}\left(\big|\sum_{l=1}^s\openone_{a_1}(Z_l) -sP(a_1)\big|\geq s\mu
P(a_1)\right)\nonumber\\
&\leq \frac{1}{s\mu^2P(a_1)}
\end{align*}
which goes to zero as $N\rightarrow \infty$ since we proved above that $sP(a_1)$ grows (exponentially) with $N$. A
similar argument shows that $P^{s} (| \hat{P}_{Z^s}(a_2)/P(a_2)-1| \geq
\mu)$ vanishes as $N$ increases. Since
$$P^{s}\left( Z^s\:\: \text{has typical type }\hat{P}\right)=P^{s}\left( \bigg|
\frac{\hat{P}_{Z^s}(a_i)}{P(a_i)}-1\bigg| <\mu,\:\:i=1,2 \right)$$
the claim is proved.


 \end{proof}

The direct part of Theorem \ref{unow} is obtained by a random coding
argument  associated with the scheme presented in Section \ref{achievability}.
We assume that all the components of all codewords
are chosen i.i.d. according
to some distribution $P$ to be specified later. Given that message $m$ starts being emitted at time $l$, we bound the probability of
error as
\begin{align*}
\pr_{m,l}(\E)\leq \pr_{m,l}(\min_{m'\ne
m}\tau_{m'}<l+N-1)+\pr_{m,l}(\tau_{m}\geq l+N)
\end{align*}
with $\tau_m$ as defined in \eqref{tdecod},
which is interpreted as the sum of the probability of false-alarm and the
probability of missing the correct codeword. In order to upper bound the above
two terms, let us define the event $E(m,n,i,k)$  as the intersection of the events
\begin{align*}
 k I(\hat{P}_{C^k(m),Y_{n-i+1}^{n-i+k}})
 +(i-k)I(\hat{P}_{C_{k+1}^i(m),Y_{n-i+k+1}^n})\geq t_2 \ln M\quad
\end{align*}
and $ \quad iD(\hat{P}_{Y_{n-i+1}^n}|| Q(\cdot|\star))\geq t_1 \ln M$.
 Also let $E(m,n,i)=\cap_{k=1,2,\ldots,i} E(m,n,i,k)$. We interpret $E(m,n,i)$
 as the event that message $m$ is declared at time $n$ by observing the
 last $i$ symbols. With these definitions we
 have\footnote{The notation $a\wedge b$ is used for the minimum of $a$ and $b$.}

\begin{align}\label{un}
\pr_{m,l}(\min_{m'\ne
m}\tau_{m'}<l+N-1)\leq \sum_{\substack{m'\ne m\\ n\in [1,\ldots,A+N-1]\\ i\in
[1,\ldots,N\wedge n] }}\pr_{m,l}\left(E(m',n,i)\right)
\end{align}
from the union bound, and
\begin{align}\label{deux}
\pr_{m,l}(\tau_{m}\geq l+N)\leq \pr_{m,l}(E(m,l+N-1,N)^c)\;.
\end{align}
Lemmas
\ref{false-alarm} and \ref{miss} below upper bound the
right hand sides of \eqref{un} and \eqref{deux}.

We denote by $\cal{P}$, $\cal{P}^X$, and $\cal{P}^Y$ the set of all distributions on
${\cal{X}}\times {\cal{Y}}$, ${\cal{X}}$, and ${\cal{Y}}$
respectively. Later we will also use ${\cal{P}}^{{\cal{Y}}|{\cal{X}}}$ to denote the set of
conditional distributions of the form $V(y|x)$ with $x\in {\cal{X}}$ and $y\in {\cal{Y}}$.
Further we denote by ${\cal{P}}_n$ the set of all types of length $n$ over $\cal{X}\times
\cal{Y}$, and similarly for ${\cal{P}}^{\cal X}_{n}$ and ${\cal{P}}^{\cal Y}_{n}$.
As mentioned earlier, the notation $\poly (N)$ is used for a term
that grows no faster than polynomially in $N$.
\begin{lem}[false-alarm]\label{false-alarm}
Assume the codebook to be randomly generated so that each sample of each codeword
is i.i.d.
according to some distribution $P$. For any threshold constants $t_1,t_2\in \mathbb{R}$ and
asynchronism level $A\geq 1$
\begin{align*}
\sum_{\substack{m'\ne m\\ n\in [1,\ldots,A+N-1]\\ i\in
[1,\ldots,N\wedge n] }}\pr_{m,l}\left(E(m',n,i)\right)&\leq
(M^{-(t_1+t_2-1)}A+M^{-(t_2-1)})\poly (N)
\;.
\end{align*}
\end{lem}
\noindent Notice that the above bound on the false-alarm error probability does not depend on $P$.
Also notice that if $t_1+t_2\leq 1$ or $t_2\leq 1$ the lemma is trivial.
\begin{proof}[Proof of Lemma \ref{false-alarm}]

 We distinguish the cases when $E(m',n,i)$ is generated outside the message
transmission period and when it is generated partly outside and partly inside
the message transmission period.
In both cases we will use the identity
\begin{align}\label{viex}
D(V||P_1P_2)= I(V)+D(V_X||P_1)+D(V_Y||P_2)\;,
\end{align}
where $V$ denotes any distribution on ${\cal{X}}\times {\cal{Y}}$ with marginals
$V_X$ and $V_Y$, and where $P_1$ and $P_2$ are any distributions on $\cal{X}$ and
$\cal{Y}$ respectively.

\emph{Case I: $E(m',n,i)$ is generated outside the message transmission period (i.e.,
$n<l$ or $n-i+1\geq l+N$)} \\
By definition $E(m',n,i) \subset E(m',n,i,i)$, hence from Theorem $12.1.4$ \cite{CT}
and \eqref{viex} we get
\begin{align}\label{ghw1}
\pr_{m,l}\left(E(m',n,i)\right)&\leq \pr_{m,l} (E(m',n,i,i))\nonumber \\
& \leq\sum_{\substack{V\in \Pc_i \\ i
    I(V)\geq t_2 \ln M \\ iD(V_Y||Q(\cdot|\star))\geq t_1 \ln M}}e^{-iD(V||PQ(\cdot|\star))}\nonumber \\ &\leq
    \sum_{\substack{V\in \Pc_i \\ i
    I(V)\geq t_2 \ln M \\ iD(V_Y||Q(\cdot|\star))\geq t_1 \ln M}}e^{-iI(V)-iD(V_Y||Q(\cdot|\star))}\nonumber \\
    &\leq
    (i+1)^{|{\cal{X}}||{\cal{Y}}|}M^{-t_2} M^{ -t_1 }\nonumber \\
    &\leq
    \poly (N)M^{-t_2} M^{ -t_1 }
    \end{align}
 where the last two inequalities hold since $|{\cal{P}}_i|\leq
(i+1)^{|{\cal{X}}||{\cal{Y}}|}$ by Lemma 2.2 \cite{CK} and because $i\leq N$.

\emph{Case II: $E(m',n,i)$ is generated partly outside and partly inside the message transmission period
(i.e., $n\geq l$ and $n-i+1\leq l+N-1$)}\\
Here the event $E(m',n,i)$ involves the output random variables $Y_{n-i+1},Y_{n-i+2},\ldots,Y_n$, the
first $k$ being distributed according to the noise
distribution, and the remaining $i-k$ according to the distribution induced by the sent codeword. Since, by definition, $E(m',n,i)\subset E(m',n,i,k)$ for any $k\in [0,1,\ldots,i]$, a
similar computation as for Case $I$ based on the identity \eqref{viex} yields
\begin{align}\label{ghw2}
\pr_{m,l}\left(E(m',n,i)\right)&\leq \pr_{m,l}\left(E(m',n,i,k)\right)\nonumber \\
&\leq \sum_{\substack{V_1\in
    \Pc_{k}, V_2 \in \Pc_{i-k} \\ k I(V_1)+
    (i-k)I(V_2)\geq t_2 \ln M
    }}e^{-kD(V_1||PQ(\cdot|\star))-(i-k)D(V_2||PP_Y)}\nonumber \\
    &\leq \sum_{\substack{V_1\in \Pc_{k}, V_2 \in \Pc_{i-k} \\
    k I(V_1)+ (i-k)I(V_2)\geq t_2 \ln M }}e^{- k
    I(V_1)- (i-k)I(V_2)}\nonumber \\ &\leq
    \poly (N)M^{-t_2}
\end{align}
where $P_Y(y)\triangleq\sum_{x\in {\cal{X}}}P(x)Q(y|x)$.

Combining the cases $I$ and $II$ we get
\begin{align*}
\sum_{\substack{m'\ne m\\ n\in [1,\ldots,A+N-1]\\ i\in
[1,\ldots,N\wedge n] }}\pr_{m,l}\left(E(m',n,i)\right)&\leq
(M^{-(t_1+t_2-1)}A+M^{-(t_2-1)})\poly (N)
\end{align*}
yielding the desired result.
\end{proof}
\begin{lem}[miss]\label{miss}
Assume the codebook to be randomly generated so that each sample of each codeword
is i.i.d.
according to some distribution $P$. For any threshold constants $t_1\geq 0$ and
$t_2\geq 0$
\begin{align*}
&\pr_{m,l}(E(m,l+N-1,N)^c)\nonumber \\
&\leq  \poly (N)\Big( \exp \big[-N\inf_{\substack{V\in {\cal{P}}^{\cal{Y}}\\
D(V||Q(\cdot|\star))<
t_1 \ln M/N}}D(V||P_Y)\big]+\exp\big[-N\min_{\substack{V
\in \Pc\\ I(V)\leq t_2 \ln M/N}} D(V||PQ)\big] \Big)
\end{align*}
where $P_Y(y)=\sum_{x\in {\cal{X}}}P(x)Q(y|x)$. (The infimum is defined to be
equal to $+\infty$ whenever the set over which it is defined is empty.).
\end{lem}
\begin{proof}
The union bound yields
\begin{align}\label{pelong}
\pr_{m,l}&(E(m,l+N-1,N)^c)\nonumber \\
&\leq \pr_{m,l}(ND(\hat{P}_{{Y^{l+N-1}_l}}||Q(\cdot|\star))<  t_1 \ln M)\nonumber \\
&+\sum_{k\in [1,\ldots,N]}\pr_{m,l}\left(
 k I(\hat{P}_{C^k(m),Y_{l}^{l+k-1}})
 +(N-k)I(\hat{P}_{C_{k+1}^N(m),Y_{l+k}^{l+N-1}})\leq t_2 \ln
 M\right)\;.
\end{align}
For the first term on the right hand side of \eqref{pelong} we get
\begin{align*}
\pr_{m,l}(ND(\hat{P}_{{Y^{l+N-1}_l}}||Q(\cdot|\star))<  t_1 \ln M)&\leq \poly (N)
\exp\big[-N\inf_{\substack{V\in {\cal{P}}^{\cal{Y}}\\
D(V||Q(\cdot|\star))<
t_1 \ln M/N}}D(V||P_Y)\big]
\end{align*}
where $P_Y(y)\triangleq\sum_{x\in {\cal{X}}}P(x)Q(y|x)$.
To prove the lemma we now show that the second term on the right hand side of \eqref{pelong} can be
bounded as
\begin{align*}
\sum_{k\in [1,\ldots,N]}\pr_{m,l}&\left(
 k I(\hat{P}_{C^k(m),Y_{l}^{l+k-1}})
 +(N-k)I(\hat{P}_{C_{k+1}^N(m),Y_{l+k}^{l+N-1}})\leq t_2 \ln
 M\right)\nonumber \\
& \leq\poly (N) \exp\big[-N\min_{\substack{V
\in \Pc\\ I(PV)\leq t_2 \ln M/N}}
D(V||PQ)\big]\;.
\end{align*}
This is done by the following inequalities
 \begin{align}\label{lautes}
\sum_{k\in [1,\ldots,N]}&\pr_{m,l}\left(
 k I(\hat{P}_{C^k(m),Y_{l}^{l+k-1}})
 +(N-k)I(\hat{P}_{C_{k+1}^N(m),Y_{l+k}^{l+N-1}})\leq t_2 \ln
 M\right)\nonumber \\
 &\leq  \sum_{\substack{V\in \Pc_{k},W\in \Pc_{N-k}\\
kI(V)+(N-k)
    I(V)\geq t_2 \ln M }} e^{-kD(V||PQ)-(N-k)D(V||PQ)} \nonumber \\
&\leq\poly (N)\exp\big[-N\min_{\delta\in[0,1]}\min_{(V,W)\in S_\delta}(\delta
 D(V||PQ)+(1-\delta)D(W||PQ))\big]\nonumber \\
&=\poly (N)\exp\big[-N\min_{V\in {\cal{P}}:I(V)\leq
\frac{t_2 \ln M}{N}}D(V||PQ)\big]
   \end{align}
   where we defined
   $$S_\delta=\{V,W\in {\cal{P}}\,: \delta I(V)+(1-\delta)I(W)\geq t_2 \ln M\}$$
and where the equality in \eqref{lautes} is justified in Lemma \ref{ppetit} given in the Appendix.
\end{proof}
\noindent The following Proposition establishes the direct part of Theorem \ref{unow} and
will be proved using Lemmas \ref{false-alarm} and~\ref{miss}.

\begin{prop}[Achievability]\label{exca}
For a channel $Q$ with strictly positive capacity, any asynchronism exponent strictly less than $$\max_x
D(Q(\cdot|x)||Q(\cdot|\star))$$ is achievable by a coding strategy
that satisfies $\lim_{N\rightarrow \infty}\ln M/N>0$.
\end{prop}
\begin{proof}
Using Lemmas \ref{false-alarm} and \ref{miss}
we get for any $A\geq 1$, $t_1\geq 0$, $t_2>1$, and distribution $P$
\begin{align}\label{ineqprop}
\pr({\EuScript{E}})\leq&\poly (N)\Bigg(M^{-(t_1+t_2-1)}A+M^{-(t_2-1)}\nonumber
\\
&+ \exp\big[-N\inf_{\substack{V\in
{\cal{P}}^{{\cal{Y}}}\\ D(V||Q(\cdot|\star))<
t_1 \ln M/N}}D(V||P_Y)\big]+\exp\big[-N\min_{\substack{V
\in {\cal{P}}\\ I(V)\leq t_2 \ln M/N}}
D(V||PQ)\big] \Bigg)
\end{align}
where $P_Y(y)=\sum_x P(x)Q(y|x)$.
We focus on the four terms
inside the large brackets of the above expression. 
For now we assume that
$Q(y|\star)>0$ for all $y\in {\cal{Y}}$, implying that
$D(P_Y||Q(\cdot|\star))<\infty$ for any input distribution $P$. The case where
$Q(y|\star)=0$ for some $y\in {\cal{Y}}$ is considered at the end of the proof. 

Pick an input distribution $P$ so that $I(PQ)>0$ and $D(P_Y||Q(\cdot|\star))>0$ (this is possible since
$C(Q)>0$), fix $t_2>1$,
 and let $\mu>0$ be a small constant (later we will take 
$t_2\rightarrow \infty$ and $\mu\rightarrow 0$).  Then choosing the ratio $\ln M/N>0$ and the constant
$t_1\geq 0$ so that
\begin{align}
\frac{t_2 \ln M}{N}=
I(PQ)-\mu/2\label{siz}
\end{align}
and 
\begin{align}\label{ase}
\frac{t_1 \ln M}{N}=
D(P_Y||Q(\cdot|\star))-\mu/2\;,
\end{align}
the second, third, and fourth term inside the large brackets in \eqref{ineqprop}
decay exponentially with $N$. Now for the first term. From \eqref{siz} and \eqref{ase} 
we get
\begin{align}\label{albet}
t_1 +t_2=\frac{N}{\ln
M}\left(D(P_Y||Q(\cdot|\star))+I(PQ)-\mu\right)\;.
\end{align}
For the first term to go to zero exponentially with $N$ we further choose $A=M^{t_1 +t_2
-(1+\mu)}$, or, equivalently using \eqref{siz} and \eqref{albet}
\begin{align}
\label{wce}
A&=e^{N\left(D(P_Y||Q(\cdot|\star))+I(PQ)-\mu-\frac{\ln
M}{N}(1+\mu)\right)}\nonumber \\
&=e^{N\left(D(P_Y||Q(\cdot|\star))+I(PQ) -\mu -\frac{1+\mu}{t_2}(I(PQ)-\mu/2)  \right)}\;.
\end{align}
Since $\mu$ can be made arbitrarily small and $t_2$ arbitrarily large
we conclude from \eqref{wce} that, as long as $A=e^{N\alpha}$ with
\begin{align}\label{gawisv}
\alpha<  D(P_Y||Q(\cdot|\star))+I(PQ)
\end{align}
the right hand side of \eqref{ineqprop} goes to zero as $N$ tends to infinity.
Maximizing the right hand side of \eqref{gawisv} over the input
distributions $P$ gives
$D(Q(\cdot|x)||Q(\cdot|\star))$, yielding the desired result. To prove this we show that\footnote{The domain
over which the supremum is taken is nonempty since $C(Q)>0$.}
\begin{align}\label{gawis}
\sup_{\substack{P\\D(P_Y||Q(\cdot|\star))>0\\I(PQ)>0}}
(D(P_Y||Q(\cdot|\star))+I(PQ))=\max_x D(Q(\cdot|x)||Q(\cdot|\star))\;.
\end{align}
Since we assumed that $Q(y|\star)>0$ for all $y\in \cal{Y}$, we have that $D(P_Y||Q(\cdot|\star))+I(PQ)$
is continuous in $P$ and therefore
$$\sup_{\substack{P\\D(P_Y||Q(\cdot|\star))>0\\I(PQ)>0}} (D(P_Y||Q(\cdot|\star))+I(PQ))=\max_P
(D(P_Y||Q(\cdot|\star))+I(PQ))\;.$$
Rewriting $D(P_Y||Q(\cdot|\star))+I(PQ)$ we get
\begin{align*}
D(P_Y||Q(\cdot|\star))+I(PQ)&=\sum_x P(x) D(Q(\cdot|x)||Q(\cdot|\star))\;,
\end{align*}
hence
\begin{align*}
\sup_{\substack{P\\D(P_Y||Q(\cdot|\star))>0\\I(PQ)>0}} (D(P_Y||Q(\cdot|\star))+I(PQ))=\max_x
D(Q(Y|x)||Q(\cdot|\star))\;.
\end{align*}
We now focus on the case where $Q(y|\star)=0$ for some $y\in {\cal{Y}}$.
Pick an input distribution $P$ such that $I(PQ)>0$ and
$D(P_Y||Q(\cdot|\star))=\infty$ --- one possibility is to take $P$ as the
uniform distribution over $\cal{X}$. Again consider the four terms into large
brackets in \eqref{ineqprop}. Fix $t_2>1$ and fix the ratio $\ln M/N$ so
that $0<\frac{t_2 \ln M}{N}<I(PQ)$.  It follows that the second and fourth term
decay exponentially with $N$. Now, with our choice of input distribution note
that the third term decays exponentially
with $N$, irrespectively of how large $t_1$ is. By letting $A=M^{t_1}$ it
follows
that the four terms decay exponentially with $N$, irrespectively of the
exponential growth
rate of $A$ with respect to $N$. Hence, when $Q(y|\star)=0$ for some $y\in
{\cal{Y}}$, an asynchronism exponent arbitrary large can be achieved.

(Note that above we always assumed $\ln M/N$ to be some strictly
positive constant. Therefore the second part of the claim of the proposition
follows.)
\end{proof}

To prove Theorem \ref{ra} we consider the same random coding argument
used in proving Proposition \ref{exca}, except that we modify the random codebook
ensemble so that each codeword now satisfies a certain prefix condition. This
condition  will allow us to treat the codewords as being essentially of constant
composition (see, e.g.,\cite[p.117]{CK}) uniformly over their length, yielding an
improved error probability exponent compared to the case where the codewords are
i.i.d. $P$. 

The random construction of a codebook satisfying the prefix condition is obtained as
follows. Given a message $m$, the codeword
$c^N(m)$ is generated so that all of its symbols are i.i.d. according to a
distribution $P$. If the obtained codeword does not satisfy the prefix condition we
discard it and regenerate a new codeword until the prefix condition is satisfied.
The prefix condition requires that all prefixes
$c^i(m)$ of size $i$ greater than $N/\ln N$ have empirical type
$\hat{P}_{c^i(m)}$ close to $P$, in the sense that
$||P-\hat{P}_{c^i(m)}||\leq  1/\ln N$.\footnote{Here $||\cdot||$ is 
the $L_1$ norm. Also, the choice $N/\ln N$ for the minimum prefix size
could be replaced by any function $f(N)$ so that
$f(N)=o(N)$ while $\ln N/f(N)=o(1)$.}
If $N$ is large enough, with overwhelming probability a random codeword will
satisfy the prefix condition. Indeed,
by the union bound, the probability of generating a sequence $c^N(m)$
that does not satisfy the prefix condition is upper bounded by
$N\exp\big[-\big|\Theta\left(N
/(\ln N)^3\right)\big|\big]$, which tends to zero as $N$ tends to infinity.  This proves
the following lemma.
\begin{lem}\label{prefix}
The probability that a sequence $C_1,C_2,\ldots,C_N$ of random variables i.i.d.
according to $P$ does not satisfy the prefix condition tends to zero as $N$ goes to infinity.
\end{lem}
\noindent To prove Theorem \ref{ra} we will need Lemmas \ref{false-alarmb} and
\ref{missb} that bound the probabilities of false-alarm and miss assuming the
codewords satisfy the prefix condition. Before establishing these lemmas
we make a small digression on the growth rate of $M$ and $N$.
Referring to the achievability scheme of Section \ref{achievability}, decoding may happen only if $i$
is so that the condition
\begin{align*}
\min_{k\in [1,\ldots,i]} \left[k I(\hat{P}_{C^k(m),Y_{n-i+1}^{n-i+k}})
+(i-k) I(\hat{P}_{C_{k+1}^i(m),Y_{n-i+k+1}^n})\right]
\geq t_2 \ln M
\end{align*}
is satisfied. Thus, a lower bound on the values of $i$ for which decoding may
happen is ${\ln M}/{\ln |{\cal{X}}|}$ since $I(\cdot)\leq \ln |{\cal{X}}|$ and $t_2>1$
. In order guarantee that, whenever decoding happens, only codeword prefixes of
size larger than $N/\ln N$ --- the size of the smallest constant composition
prefix --- are involved we impose that $M$ and $N$ satisfy
%
\begin{align}\label{sca}
 \frac{N}{\ln N}\leq \frac{\ln M}{\ln  |{\cal{X}}|}\;.
\end{align}
\begin{lem}[false-alarm, with prefix condition]\label{false-alarmb}
Assume the codebook to be randomly generated so that  each codeword satisfies the prefix
condition according to $P$, and assume that \eqref{sca} holds. For any threshold constants $t_1,t_2\in \mathbb{R}$ and any asynchronism
level $A\geq 1$
\begin{align*}
\sum_{\substack{m'\ne m\\ n\in [1,\ldots,A+N-1]\\ i\in
[1,\ldots,N\wedge n] }}\pr_{m,l}\left(E(m',n,i)\right)&\leq
\poly (N)(M^{-(t_1+t_2-1+o(1))}A+M^{-(t_2-1+o(1))})
\end{align*}
as $N\rightarrow \infty$.
\end{lem}

\begin{lem}[miss, with prefix condition]\label{missb}
Assume the codebook to be randomly generated so that  each codeword satisfies the prefix
condition according to $P$ and assume that \eqref{sca} holds. For any $t_1\geq
0$ and  $t_2>0$
\begin{align}
\pr_{m,l}&(E(m,l+N-1,N)^c)\nonumber \\
&\leq  \poly (N)\Big( \exp\big[-N\inf_{\substack{V\in
{\cal{P}}^{{\cal{Y}}|{\cal{X}}}\\ D((PV)_Y||Q(\cdot|\star))<
t_1 \ln M/N}}D((PV)_Y||P_Y)(1+o(1))\big]\nonumber \\
&\hspace{1.78cm}+\exp\big[-N\min_{\substack{V
\in {\cal{P}}^{{\cal{Y}}|{\cal{X}}}\\ I(PV)\leq t_2 \ln M/N}}
D(PV||PQ)(1+o(1))\big] \Big)\label{sequel}
\end{align}
as $N\rightarrow \infty$, where $P_Y(y)=\sum_{x\in {\cal{X}}}P(x)Q(y|x)$.
\end{lem}

Comparing Lemma \ref{false-alarm} with Lemma \ref{false-alarmb} and Lemma
\ref{miss} with Lemma \ref{missb} we see that the false-alarm probability
bounds are essentially the same with and without the prefix condition, whereas
for the miss probability the bound is improved by the prefix condition. Note
also that, for the miss probability, the bound obtained  with the prefix
condition is the sum of two terms that involve convex optimizations, whereas
the bound without the prefix condition involves a non convex optimization, in
general more difficult to handle. To prove Lemmas \ref{false-alarmb} and
\ref{missb} we use similar bounding techniques as in the proofs of Lemmas
\ref{false-alarm} and \ref{miss} together with the following argument.

Suppose $\{(C_i,Y_i)\}_{i=1,\ldots,n}$ is a
sequence of i.i.d. pairs of random variables taking values in ${\cal{X}}\times
\cal{Y}$ so that $(C_1,Y_1)$ is distributed
according to some $J\in \cal{P}$.  It then follows, by Theorem \cite[Theorem
12.1.4]{CT}, that for a given type $V=V_XV_{Y|X}$ in ${\cal{P}}_n$
\begin{align}
\label{sspref}
\pr((C^n,Y^n) \:\; \text{has type}\;\; V) \leq  e^{-nD(V_XV_{Y|X}||J)}\;,
\end{align}
which implies that
\begin{align}\label{alep}
\pr((C^n,Y^n) \:\; \text{has type}\;\; V\;|\;C^n &\;\text{satisfies prefix
 condition})\pr(C^n \;\text{satisfies prefix
 condition})\nonumber \\
 &\leq e^{-nD(V_XV_{Y|X}||J)}\;.
 \end{align}
Now assuming that $n$ is larger than $N/\ln N$, the size of the smallest
codeword length that satisfies the prefix condition, we have that
$$\pr((C^n,Y^n) \:\; \text{has type}\;\; V\;|\; C^n \; \text{satisfies the prefix
condition})$$
has nonzero probability only if $||V_X -P||\leq 1/\ln N$. Assuming so,
since the probability that $C^n$ satisfies the prefix condition tends to one as
$n\rightarrow \infty$
(Lemma \ref{prefix}) we conclude from \eqref{alep} and by continuity of
$D(\cdot||J)$ that
 \begin{align}\label{avecpref}
 \pr((C^n,Y^n) \:\; \text{has type}\;\; V\;|\;C^n \;\text{satisfies prefix
 condition}){\leq} e^{-nD(PV_{Y|X}||J)(1+o(1))}
 \end{align}
as $N\rightarrow \infty$.

Comparing \eqref{sspref} and \eqref{avecpref} we see that the prefix condition 
essentially allows us to treat $C^n$ as being of composition $P$. Accordingly,
to prove Lemmas \ref{false-alarmb} and \ref{missb} we follow the steps of the proofs of
Lemmas \ref{false-alarm} and \ref{miss} and repeatedly use the above argument
(without explicitly mentioning it everywhere) in order to incorporate the prefix condition and change
the large deviations exponent of the form $D(V_XV_{Y|X}||J)$ to
$D(PV_{Y|X}||J)$.  The only additional technicality relates to the
small discrepancy that occurs because the prefix condition does not hold for
small prefix lengths, i.e., lengths smaller than $N/\ln N$. We recall that $M$ and $N$ are assumed to satisfy
\eqref{sca}.

\begin{proof}[Proof of Lemma \ref{false-alarmb}]

\emph{Case I: $E(m',n,i)$ is generated outside the message transmission period (i.e.,
$n<l$ or $n-i+1\geq l+N$)} \\
A similar computation as in \eqref{ghw1} yields as $N\rightarrow \infty$
\begin{align*}
\pr_{m,l}\left(E(m',n,i)\right)&\leq \pr_{m,l} (E(m',n,i,i))\nonumber \\
& \leq \sum_{\substack{V\in \Pc_i, V_X\approx P  \\ i
    I(V)\geq t_2 \ln M \\ iD(V_Y||Q(\cdot|\star))\geq t_1 \ln
    M}}e^{-iD(V||PQ(\cdot|\star))(1+o(1))}\nonumber \\ &\leq
    \sum_{\substack{V\in \Pc_i \\ i
    I(V)\geq t_2 \ln M \\ iD(V_Y||Q(\cdot|\star))\geq t_1 \ln
    M}}e^{-i(I(V)+D(V_Y||Q(\cdot|\star))(1+o(1))}\nonumber \\ &\leq
    \poly (N)M^{-t_2-t_1+o(1) }\;.
\end{align*}
where $V_X\approx P$ denotes $||V_X-P||\leq 1/\ln N$.

\emph{Case II: $E(m',n,i)$ is generated partly outside and partly inside the message transmission period
(i.e., $n\geq l$ and $n-i+1\leq l+N-1$)}\\
The event $E(m',n,i)$ involves the output random variables $Y_{n-i+1},Y_{n-i+2},\ldots,Y_n$, the
first $k$ being distributed according to the noise
distribution, and the remaining $i-k$ according to the distribution induced by the sent codeword. 
In order to deal with the discrepancy that results because codeword lengths of size
smaller than $N/\ln N$ do not satisfy the prefix condition, we distinguish two cases.
\begin{itemize}
\item
$k\geq N/\ln N$ and $i-k\geq N/\ln N$\\
A similar computation as in \eqref{ghw2} yields
\begin{align*}
\pr_{m,l}\left(E(m',n,i)\right)&\leq \sum_{\substack{V\in
    \Pc_{k}, W\in \Pc_{i-k} \\ V_X=P\approx \varepsilon, W_X\approx P\pm \varepsilon \\ k I(V)+
    (i-k)I(W)\geq t_2 \ln M
    }}e^{-(kD(V_1||PQ(\cdot|\star))+(i-k)D(V_2||PP_Y))(1+o(1))}\nonumber \\
    &\leq \sum_{\substack{V_1\in \Pc_{k}, V_2 \in \Pc_{i-k} \\
    k I(V_1)+ (i-k)I(V_2)\geq t_2 \ln M }}e^{- (k
    I(V_1)+(n-i)I(V_2))(1+o(1))}\nonumber \\ &\leq
    \poly (N)M^{-\alpha+o(1)}
\end{align*}
where $P_Y(\cdot)\triangleq\sum_{x\in {\cal{X}}}P(x)Q(\cdot|x)$.

\item
$k< N/\ln N$ or $i-k< N/\ln N$\\
We consider only the case $k< N/\ln N$, the case $i-k< N/\ln N$ being
obtained in the same way. Since $I(V)\leq \ln |{\cal{X}}|$ we have as $N\rightarrow \infty$
\begin{align*}
\pr_{m,l}\left(E(m',n,i)\right)&\leq \sum_{\substack{ V\in \Pc_{i-k}, V_X\approx
P
    \\ (N/\ln N)\ln |{\cal{X}}|+(i-k)I(V)\geq t_2 \ln M
    }}e^{-(i-k)(D(V||PP_Y)(1+o(1))}\nonumber \\
    &\leq \sum_{\substack{V\in \Pc_{i-k},V_X\approx P \\
 (N/\ln N)\ln|{\cal{X}}|+   (i-k)I(V)\geq t_2 \ln M }}e^{-
 (i-k)I(V)(1+o(1))}\nonumber \\ &\leq
    \poly (N)M^{-t_2+o(1)}\;.
\end{align*}

\end{itemize}
Combining the cases $I$ and $II$ we get as $N\rightarrow \infty$
\begin{align*}
\sum_{\substack{m'\ne m\\ n\in [1,\ldots,A+N-1]\\ i\in
[1,\ldots,N\wedge n] }}\pr_{m,l}\left(E(m',n,i)\right)&\leq
(M^{-(t_2-t_1-1+o(1))}A+M^{-(t_2-1+o(1))})\poly (N)
\end{align*}
yielding the desired result.
\end{proof}
\begin{proof}[Proof of Lemma \ref{missb}]
According to the proof of Lemma \ref{miss} we need to bound 
$$\pr_{m,l}(ND(\hat{P}_{{Y^{l+N-1}_l}}||Q(\cdot|\star))<  t_1 \ln M)$$
and
$$\sum_{k\in [1,\ldots,N]}\pr_{m,l}\left(
 k I(\hat{P}_{C^k(m),Y_{l}^{l+k-1}})
 +(N-k)I(\hat{P}_{C_{k+1}^N(m),Y_{l+k}^{l+N-1}})\leq t_2 \ln
 M\right)\;.$$
For the first term we apply the argument that precedes Lemma \ref{false-alarmb}
and immediately obtain
\begin{align}\label{pourap}
\pr_{m,l}&(ND(\hat{P}_{{Y^{l+N-1}_l}}||Q(\cdot|\star))<  t_1 \ln M)\nonumber \\
&\leq \poly (N) 
\exp\big[-N\inf_{\substack{V\in
{\cal{P}}^{{\cal{Y}}|{\cal{X}}}\\ D((PV)_Y||Q(\cdot|\star))<
t_1 \ln M/N}}D((PV)_Y||P_Y)(1+o(1))\big]
\end{align}
as $N\rightarrow \infty$. For the second term we proceed along the lines of
the set of inequalities \eqref{lautes} and, similarly to the case II
in the proof of Lemma \ref{false-alarmb}, we separately consider the situations
$k<N/\ln N $ and $k\geq N/\ln N$.  This yields
\begin{align*}
\sum_{k\in [1,\ldots,N]}&\pr_{m,l}\left(
 k I(\hat{P}_{C^k(m),Y_{l}^{l+k-1}})
 +(N-k)I(\hat{P}_{C_{k+1}^N(m),Y_{l+k}^{l+N-1}})\leq t_2 \ln
 M\right)\\
 &\leq \poly (N)\exp\big[-N\min_{\substack{V
\in {\cal{P}}^{{\cal{Y}}|{\cal{X}}}\\ I(PV)\leq t_2 \ln M/N}}
D(PV||PQ)(1+o(1))\big]
\end{align*}
as $N\rightarrow \infty$, which concludes the proof.
\end{proof}

\begin{proof}[Proof of Theorem \ref{ra}]
The proof is obtained by deriving bounds on the average decoding delay $(\tau-\nu)^+$
and on the error probability event
$\EuScript{E}$. In what follows we assume that the ratio 
$\ln M/N$ remains fixed as $N \rightarrow \infty$ so that \eqref{sca} is satisfied.
This in turn allow us to use Lemmas \ref{false-alarmb} and \ref{missb}. Also,
from now on we assume that $P$ is so that $I(PQ)>0$.

\noindent The average decoding delay is bounded as
\begin{align}\label{esppx}
\ex_{m,l}(\tau-l)^+&\leq \ex_{m,l}(\tau_m-l)^+\nonumber \\
&=\ex_{m,l}(\openone_{\tau_m<
l+N}(\tau_m-l)^+)+\ex_{m,l}(\openone_{\tau_m\geq l+N}(\tau_m-l)^+)
\end{align}
where $\openone_{\tau_m\geq l+N}$ is equal one if $\tau_m\geq l+N$, zero else. 

For the first term on the right hand side of \eqref{esppx} we have
\begin{align}\label{shitz}
\ex_{m,l}(\openone_{\tau_m< l+N}(\tau_m-l)^+)\leq j+ N\pr_{m,l}(\tau_m\geq
l+j)\;,
\end{align}
where\footnote{The term $1/M$ in the definition of $j$ can be replaced by any positive strictly
decreasing function of $M$.} $$j\triangleq \frac{t_2 \ln M(1+1/M)}{I(PQ)}d(\delta)\;,$$ 
with
\begin{align}\label{dm}
d(\delta)\triangleq \frac{I(PQ)}{\min_{\substack{V\in {\cal{P}}^{{\cal{Y}}|{\cal{X}}}
\\D(PV||PQ)\leq \delta}}  I(PV)}
\end{align}
and $\delta=\delta(M)=1/\sqrt{\ln M}$.
For now we assume that
\begin{align}\label{les}
j=\frac{t_2\ln M}{I(PQ)}(1+o(1))\quad \text{as } N\rightarrow \infty
\end{align} 
and show that the term $N\pr_{m,l}(\tau_m\geq
l+j)$ goes to zero as $N$ tends to infinity --- the equality \eqref{les} will be shown
at the end of the proof.
  Using the inequality
\eqref{sequel} with $N$ replaced by $j$ yields
\begin{align}
\label{dfwww}
{\mathbb{P}}_{m,l}(\tau_m\geq l+j)&\leq \pr_{m,l} \left(E(m,l+j-1,j)^c
\right)\nonumber \\
&\leq  \poly (N)\Big(\exp\big[-j\min_{\substack{V
\in {\cal{P}}^{{\cal{Y}}|{\cal{X}}}\\ I(PV)\leq t_2 \ln M/j}}
D(PV||PQ)(1+o(1))\big]\nonumber \\
&\hspace{1.86cm}+\exp\big[-j\inf_{\substack{V\in
{\cal{P}}^{{\cal{Y}}|{\cal{X}}}\\ D((PV)_Y||Q(\cdot|\star))<
t_1 \ln M/j}}D((PV)_Y||P_Y)(1+o(1))\big] \Big)\;.
\end{align}
We evaluate the first term in the large brackets in \eqref{dfwww}. Expanding $d(\delta)$ in
the definition of $j$ we get
\begin{align}\label{noni}\frac{t_2 \ln M(1+1/M)}{j}=\min_{V\in {\cal{P}}^{{\cal{Y}}|{\cal{X}}}:
D(PV||PQ)\leq
\delta}I(PV)
\end{align}
implying that\footnote{Here we are using the fact that if for some $\varepsilon
>0$ we have $\min_{x:
g(x)\leq c}f(x)=m+\varepsilon$, then
$\min_{x:f(x)\leq m}g(x)\geq c$.}
\begin{align*}
\min_{V\in {\cal{P}}^{{\cal{Y}}|{\cal{X}}}: I(PV)\leq
\frac{t_2 \ln M}{j}} D(PV||PQ)\geq
\delta\;.
\end{align*}
Since $\delta=1/\sqrt{\ln M}$ we obtain
\begin{align}
\label{dfw}
\exp\big[-j\min_{\substack{V
\in \Pc_{Y|X}\\ I(PV)\leq \frac{t_2 \ln M}{j}}} D(PV||PQ)\big]&\leq e^{-\Theta(\sqrt{\ln M})}\;.
\end{align}
We now turn to the second term in the large brackets in \eqref{dfwww}. Since
$j=\frac{t_2 \ln M}{I(PQ)}(1+o(1))$, we assume that $P$, $t_1\geq 0$, and
$t_2>1$ satisfy 
\begin{align}\label{enebo}
t_1< \frac{t_2
D(P_Y||Q(\cdot|\star))}{I(PQ)}
\end{align}
 so that $$\inf_{\substack{V\in
{\cal{P}}^{{\cal{Y}}|{\cal{X}}}\\ D((PV)_Y||Q(\cdot|\star))<
t_1 \ln M/j}}D((PV)_Y||P_Y)>0\;,$$
and hence
\begin{align}\label{dfw1}
\exp\big[-j\inf_{\substack{V\in
{\cal{P}}^{{\cal{Y}}|{\cal{X}}}\\ D((PV)_Y||Q(\cdot|\star))<
t_1 \ln M/j}}D((PV)_Y||P_Y)\big]\leq e^{-\Theta (\ln M)}\;.
\end{align}
From \eqref{dfwww}, \eqref{dfw}, and \eqref{dfw1} we have 
\begin{align*}
N\pr_{m,l}(\tau_m\geq
l+j)\rightarrow 0\quad \text{as }N\rightarrow \infty\;,
\end{align*}
and using \eqref{shitz} and \eqref{les} it follows that
\begin{align}\label{loublie}
\ex_{m,l}(\openone_{\tau_m< l+N}(\tau_m-l)^+)\leq \frac{t_2\ln M}{I(PQ)}(1+o(1))\;.
\end{align}

For the second
term on the right hand side of the equality in \eqref{esppx} we get
\begin{align*}
\ex_{m,l}(\openone_{\tau_m\geq l+N}(\tau_m-l)^+)\leq (A+N)\pr_{m,l}(\tau_m\geq
l+N)
\end{align*}
since $\tau_m\leq A+N-1$.
Further, using Lemma \ref{missb}
 \begin{align*}
\pr_{m,l}(\tau_m\geq l+N)&\leq \pr_{m,l}(E(m,l+N-1,N)^c)\nonumber \\
\leq  \poly (N)\Big( &\exp\big[-N\inf_{\substack{V\in
{\cal{P}}^{{\cal{Y}}|{\cal{X}}}\\ D((PV)_Y||Q(\cdot|\star))<
t_1 \ln M/N}}D((PV)_Y||P_Y)(1+o(1))\big]\nonumber \\
+&\exp\big[-N\min_{\substack{V
\in {\cal{P}}^{{\cal{Y}}|{\cal{X}}}\\ I(PV)\leq t_2 \ln M/N}}
D(PV||PQ)(1+o(1))\big] \Big)
\;,
\end{align*}
and thus
\begin{align}\label{biewx}
\ex_{m,l}&(\openone_{\tau_m\geq l+N}(\tau_m-l)^+)\nonumber \\
&\leq \poly (N)A\Big( \exp\big[-N\inf_{\substack{V\in
{\cal{P}}^{{\cal{Y}}|{\cal{X}}}\\ D((PV)_Y||Q(\cdot|\star))<
t_1 \ln M/N}}D((PV)_Y||P_Y)(1+o(1))\big]\nonumber \\
&\hspace{2.15cm}+\exp\big[-N\min_{\substack{V
\in {\cal{P}}^{{\cal{Y}}|{\cal{X}}}\\ I(PV)\leq t_2 \ln M/N}}
D(PV||PQ)(1+o(1))\big]
\Big)\;.
\end{align}
Letting $A=e^{N\alpha}$ with $\alpha\geq 0$ we have
$$\ex_{m,l}(\openone_{\tau_m\geq l+N}(\tau_m-l)^+)=o(1)\quad \text{as
}N\rightarrow \infty$$
provided that
 $P$, $t_1\geq 0$, $t_2> 1$, and the ratio $\ln M/N$ can be chosen so that the
inequalities
\begin{align}\label{baril}
\alpha &<\inf_{\substack{V\in {\cal{P}}^{{\cal{Y}}|{\cal{X}}}\\ D((PV)_Y||Q(\cdot|\star))\leq
t_1 \ln M/N}}D((PV)_Y||(PQ)_Y)\nonumber \\
 \alpha &<\min_{\substack{V
\in {\cal{P}}^{{\cal{Y}}|{\cal{X}}}\\ I(PV)\leq \frac{t_2 \ln M}{N}}} D(PV||PQ)
\end{align}
are satisfied. Therefore, if the inequalities from \eqref{enebo} and
\eqref{baril} are satisfied the delay is bounded as
\begin{align}\label{igze}\ex_{m,l}(\tau_m-l)^+\leq \frac{t_2\ln
M}{I(PQ)}(1+o(1))\;.\end{align}

We now bound the error probability. To that aim we consider the false-alarm
and miss events and obtain, by Lemmas \ref{false-alarmb} and \ref{missb}
\begin{align}\label{aneqprop}
\pr({\EuScript{E}})\leq\poly (N)\Bigg(&M^{-(t_1+t_2-1)(1+o(1))}A+M^{-(t_2-1)(1+o(1))}\nonumber
\\
&+ \exp\big[-N\min_{\substack{V\in
{\cal{P}}^{{\cal{Y}}|{\cal{X}}}\\ D((PV)_Y||Q(\cdot|\star))<
t_1 \ln M/N}}D((PV)_Y||(PQ)_Y)(1+o(1))\big]\nonumber \\
&+\exp\big[-N\inf_{\substack{V
\in {\cal{P}}^{{\cal{Y}}|{\cal{X}}}\\ I(PV)\leq t_2 \ln M/N}}
D(PV||PQ)(1+o(1))\big] \Bigg)\;.
\end{align}
Therefore, if in addition to the three inequalities given in \eqref{enebo} and
\eqref{baril} we impose that the ratio $\ln M/N$ satisfies
\begin{align*}
\frac{\ln M}{N}\geq \frac{\alpha}{\delta(t_1+t_2-1)}
\end{align*}
for some $\delta\in (0,1)$, the right hand side of \eqref{aneqprop} goes to zero as
$N$ tends to infinity, and using \eqref{igze} we deduce that the asynchronism exponent $\alpha$ can be achieved at rate $I(PQ)/t_2$.

To summarize, if $P$, $t_1\geq 0$, $t_2>1$, $\alpha$, and the ratio $\ln M/N$ satisfy the following conditions
\begin{align}\label{rmountain}
a.\quad &\alpha <\inf_{\substack{V\in {\cal{P}}^{{\cal{Y}}|{\cal{X}}}\\ D((PV)_Y||Q(\cdot|\star))<
\frac{t_1\alpha}{\delta(t_1+t_2-1)}}}D((PV)_Y||(PQ)_Y)\nonumber \\
b.\quad &\alpha <\min_{\substack{V
\in {\cal{P}}^{{\cal{Y}}|{\cal{X}}}\\ I(PV)\leq \frac{t_2\alpha}{\delta(t_1+t_2-1)}}} D(PV||PQ)\nonumber \\
c.\quad &\frac{t_1}{t_2}<
\frac{D((PQ)_Y||Q(\cdot|\star))}{I(PQ)}\\
d.\quad  &\frac{\ln M}{N}\geq \frac{\alpha}{\delta(t_1+t_2-1)}
\end{align}
for some $ \delta\in (0,1)$, then the
asynchronism exponent $\alpha$ can be achieved at rate $I(PQ)/t_2$. Note that
if the conditions $a$, $b$, and $c$ are satisfied for some $\alpha$,
$P$, $t_1\geq 0$, 
$t_2>1$, and $\delta\in (0,1)$ one can always find choose $N/\ln M$ so that the condition $d$ 
is satisfied. Hence, if the conditions $a$, $b$, and $c$ are satisfied
for some $\alpha$, $P$, $t_1\geq 0$, $t_2>1$, and $\delta\in (0,1)$ the
asynchronism exponent $\alpha$ can be achieved at rate $I(PQ)/t_2$.

To conclude the proof we show that $j=\frac{t_2 \ln M}{I(PQ)}(1+o(1))$. To that aim we show that $d(\delta)=1+o(1)$ as $\delta \rightarrow
0$. Since $I(PV)$ is a continuous function over the compact set
\begin{align}\label{setens}
\{V\in {\cal{P}}^{{\cal{Y}}|{\cal{X}}}:D(PV||PQ)\leq \delta\}\;,
\end{align}
 the minimum in the denominator of the right hand side of
(\ref{dm}) is well defined, and so is $d(\delta)$. We now show that for $\delta$
small enough, the set in \eqref{setens} contains no trivial conditional
probability $V$, that is no $V\in {\cal{P}}^{{\cal{Y}}|{\cal{X}}}$ such that
$V(\cdot|x)$ is the same for all $x\in {\cal{X}}$. This will imply that $d(\delta)=1+o(1)$ as $\delta \rightarrow 0$.

Let $W(x,y)=W_X(x)W_Y(y)$ for all $(x,y)\in {\cal{X}}\times {\cal{Y}}$. The
identity \eqref{viex} yields
\begin{align}\label{dm2}
D(PQ||W)&= I(PQ)+D(P||W_X)+D(P_Y||W_Y)\nonumber \\
&\geq I(PQ)
\end{align}
where $P_Y(y)\triangleq\sum_{x\in {\cal{X}}}P(x)Q(y|x)$. Since
the set ${\cal{P}}^\pi$ of product measures in ${\cal{P}}$ is compact
and $D(PQ||\,\cdot\,)$ is continuous over ${\cal{P}}^\pi$, from
(\ref{dm2}) we have
\begin{align}\label{dm4}
\min_{W\in {{\cal{P}}^\pi}}
D(PQ||W)\geq I(PQ)\;.
\end{align}
Since  $I(PQ)>0$, from \eqref{dm4} one deduces that $\min_{W\in {\cal{P}}^\pi} D(W||PQ)$ is strictly
positive\footnote{We use the fact that $D(P_1||P_2)=0$ if and only if $P_1=P_2$. }
and therefore the set \eqref{setens}
contains no trivial conditional probability. Therefore,
for $\delta$ small enough the denominator in the definition \eqref{dm}
is strictly positive, implying that $d(\delta)$ is
finite. We then deduce that $d(\delta)=1+o(1)$ as $\delta \rightarrow 0$.
\end{proof}

\section{Concluding remarks}
\label{conclusione}
 We introduced a new model for asynchronous and sparse communication
and derived scaling laws between asynchronism level and blocklength for
reliable and quick decoding. Perhaps the main conclusion is that even in the
regime of strong asynchronism, i.e., when the asynchronism level is exponential
with respect to the codeword length, reliable and quick decoding can be
achieved.

 At this point several directions might be pursued. Perhaps the first is the
characterization of the asynchronism exponent function $\alpha(\cdot,Q)$ at
positive rates. In order to make this problem easier one may want to consider a
less stringent rate definition. Indeed, the definition of rate we adopted
considers $\ex(\tau-\nu)^+$ as delay. As a consequence, in the exponential
asynchronism level we mostly focused on, it is difficult to guarantee high
communication rate; even though the probability of `missing the codeword' is
exponentially small in the codeword length, once the codeword is missed we pay
a huge penalty in terms of delay,  of the order of the asynchronism level which
is exponentially large in the codeword length. Therefore, instead of imposing $\ex(\tau-\nu)^+$
to be bounded by some $d$, we may consider a delay constraint of the
form $\pr((\tau-\nu)^+\leq d)\approx 1$ and define the rate as $\ln M/d$. 

Another direction is the extension of the proposed model to include the event
when no message is sent; the receiver knows that with  probability $1-p$ one
message is sent and with probability $p$ no
message is sent. For this setting `natural' scalings between $p$ and the
asynchronism level remain to be discovered.

Finally a word about feedback. We omitted feedback in our study in order to avoid a potential
additional source of asynchronism. Nevertheless since feedback is inherently
available in any communication system it is of interest to include, say, a
one-bit perfect feedback from the receiver to the transmitter. In this case
variable length codes can be used and the asynchronism level might
be defined directly with respect to $\ex(\tau-\nu)^+$ instead of the blocklength. 

\section{Appendix}
\begin{proof}[Proof of Lemma \ref{bucket}]
The binomial expansion for $P^s(T({{\hat{P}}}))$ (see, e.g., \cite[equation 12.25]{CT}) gives
$$P^s(T({{\hat{P}}}))=\binom{s}{s{{\hat{P}}}(a_1),s{{\hat{P}}}(a_2),\ldots,s{{\hat{P}}}(a_{|{\cal{A}}|})}\prod_{a\in
{\cal{A}}}P(a)^{s{{\hat{P}}}(a)}\;.$$ 
Using the hypothesis on
$P$, ${\hat{P}}$, and $\bar{P}$ gives $\hat{P}(a_i)\geq 3/s$, $i\in \{1,2\}$,
hence
\begin{align*}
\frac{P^s(T(\bar{P}))}{P^s(T({{\hat{P}}}))}&=\left(\frac{P(a_2)}{P(a_1)}\right)^3\frac{(s{\hat{P}}(a_1)-2)(s{\hat{P}}(a_1)-1)(s{\hat{P}}(a_1))}{(s{\hat{P}}(a_2)+1)(s{\hat{P}}(a_2)+2)(s{\hat{P}}(a_2)+3)}\nonumber
\\
&=\left(\frac{P(a_2)}{P(a_1)}\right)^3\left(\frac{{\hat{P}}(a_1)}{{\hat{P}}(a_2)}\right)^3
\frac{(1-1/s{\hat{P}}(a_1))(1-2/s{\hat{P}}(a_1))}{(1+1/s{\hat{P}}(a_2))(1+2/s{\hat{P}}(a_2))(1+3/s{\hat{P}}(a_2))}\\
&\geq \delta
\end{align*}
for some $\delta=\delta(\delta_0)>0$.
\end{proof}

\begin{lem}
\label{ppetit}
For any distribution $J$ on $\cal{X}\times {\cal{Y}}$ and any constant $r\geq 0$
\begin{align*}
\min_{t_1\in [0,1]}\min_{\substack{V_1,V_2\in {\cal{P}}\\
t_1 I(V_1)+(1-t_1)I(V_2)\leq r}}t_1 D(V_1||J)+(1-{t_1})D(V_2||J)=
\min_{\substack{V\in {\cal{P}}\\
I(V)\leq r}} D(V||J)\;.
\end{align*}
\end{lem}
\begin{proof}
If $r\geq I(J)$ the claim trivially holds, since the left and right hand side of the above equation equal to
zero. From now on we assume that $r<I(J)$.

Define
\begin{align*}
a=\min_{t_1\in [0,1]}\min_{\substack{V_1,V_2\in {\cal{P}}\\
t_1 I(V_1)+(1-t_1)I(V_2)\leq r\\ I(V_1)=I(V_2)}}t_1 D(V_1||J)+(1-{t_1})D(V_2||J)
\end{align*}
and
\begin{align*}
b=\min_{t_1\in [0,1]}\inf_{\substack{V_1,V_2\in {\cal{P}}\\
t_1 I(V_1)+(1-t_1)I(V_2)\leq r\\ I(V_1)>I(V_2)}}t_1
D(V_1||J)+(1-{t_1})D(V_2||J)\;.
\end{align*}
Since $a=\min_{\substack{V\in {\cal{P}}\\
I(V)\leq r}} D(V||J)$ to prove the Lemma it suffices to show that
$b\geq \min_{\substack{V\in {\cal{P}}\\
I(V)\leq r}} D(V||J)$. This is done via the following two claims proved below:
\begin{itemize}
\item
claim i. $\min_{V:I(V)\leq r}D(V||J)=\min_{V:I(V)= r}D(V||J)$.
\item
claim ii. the function $f(r)\triangleq \min_{V:I(V)=r}D(V||J)$ is convex.
\end{itemize}
\noindent Using the above claims we have
\begin{align*}
b&=\inf_{\substack{r_1>r_2\\ \frac{r-r_2}{r_1-r_2}r_2+\frac{r_1-r}{r_1-r_2}r_1=r}}
\frac{r-r_2}{r_1-r_2}f(r_1)+\frac{r_1-r}{r_1-r_2}f(r_2)\nonumber \\
&\geq f(r)
\end{align*}
and therefore $b\geq \min_{\substack{V\in {\cal{P}}\\
I(V)\leq r}} D(V||J)$.

The proof of the above claims is based on the convexity of $D(J_1||J_2)$ in the
pair $(J_1,J_2)$ (see,
e.g., \cite[Lemma 3.5, p.50]{CK}). For claim i, let $r>0$ and suppose that
$I(V)<r$.\footnote{If $r=0$ the claim holds trivially.} By  defining $\bar{V}=\lambda
V+(1-\lambda)J$ with $\lambda \in [0,1)$ we have
$D(\bar{V}||J)<D(V||J)$ by convexity.
On the other hand letting $V_X$ and $V_Y$
denote the left and right marginals of $V$ we have
we have
\begin{align*}
I(\bar{V})&=D(\lambda V+(1-\lambda)J||\lambda V_XV_Y+(1-\lambda)J_XJ_Y)\\
& =\lambda D( V|| V_XV_Y)+(1-\lambda)D(J||J_XJ_Y)\\
&=\lambda I(V)+(1-\lambda)I(J)\\
&<r
\end{align*}
where the inequality holds for $\lambda$ sufficiently close to one. Therefore
$\bar{V}$ strictly improves upon $V$ and claim i follows.\footnote{Notice that in
\cite[p.169]{CK} a similar argument holds for the sphere packing exponent.}

For claim ii, let $V_1$ and $V_2$ achieve $f(r_1)$ and $f(r_2)$, for some $r_1\ne
r_2$, and let $V=\lambda V_1+(1-\lambda)V_2$. By convexity we have
\begin{align*}
D(V||J)&\leq \lambda D(V_1||J)+(1-\lambda)D(V_2||J)\\
&=\lambda f(r_1)+(1-\lambda)f(r_2)
\end{align*}
and $I(V)\leq r$. This yields claim ii.
\end{proof}

\section*{Acknowledgment}
The authors wish to thank Ashish Khisti for interesting
discussions.

\bibliographystyle{amsplain}
\bibliography{../../../common_files/bibiog}

\providecommand{\bysame}{\leavevmode\hbox to3em{\hrulefill}\thinspace}
\providecommand{\MR}{\relax\ifhmode\unskip\space\fi MR }
\providecommand{\MRhref}[2]{%
  \href{http://www.ams.org/mathscinet-getitem?mr=#1}{#2}
}
\providecommand{\href}[2]{#2}
\begin{thebibliography}{10}

\bibitem{AW}
R.~Ahlswede and J.~Wolfowitz, \emph{Channels without synchronization}, Advances
  in applied probability \textbf{3} (1971), no.~2, 383--403.

\bibitem{BN}
Mich\`ele Basseville and Igor Nikiforov, \emph{Fault isolation for diagnosis:
  nuisance rejection and multiple hypothesis testing}, rapport de recherche
  4438, INRIA, 2002.

\bibitem{CMP}
T.~M. Cover, R.~J. McEliece, and E.~C. Posner, \emph{Asynchronous
  multiple-access channel capacity}, IEEE Trans.~Inform.~Th. \textbf{4} (1981),
  409--413.

\bibitem{CT}
T.M. Cover and J.A. Thomas, \emph{Elements of information theory}, Wiley, New
  York, 1991.

\bibitem{CK}
I.~Csisz\`ar and J.~K\"orner, \emph{Information theory: Coding theorems for
  discrete memoryless channels}, Academic Press, New York, 1981.

\bibitem{DG}
S.~Diggavi and M.~Grossglauser, \emph{On transmission over deletion channels},
  Allerton Conference, Monticello, Illinois, October, 2001.

\bibitem{D2}
R.~L. Dobrushin, \emph{Shannon's theorems for channels with synchronization
  errors}, Problems Information transmission \textbf{3} (1967), no.~4, 11--26.

\bibitem{DM3}
E.~Drinea and M.~Mitzenmacher, \emph{A simple lower bound for the capacity of
  the deletion channel}, IEEE Trans.~Inform.~Th. \textbf{52} (2006),
  4657--4660.

\bibitem{HuH}
J.Y.N. Hui and P.A. Humblet, \emph{The capacity of the totally asynchronous
  multiple-access channel}, IEEE Trans.~Inform.~Th. \textbf{31} (1985), no.~2,
  207--216.

\bibitem{Lai2}
T.Z. Lai, \emph{Sequential multiple hypothesis testing and efficient fault
  detection-isolation in stochastic systems}, IEEE Trans.~Inform.~Th.
  \textbf{46} (2000), 595--607.

\bibitem{N}
I.~V. Nikiforov, \emph{A generalized change detection problem}, IEEE
  Trans.~Inform.~Th. \textbf{41} (1995), 171--187.

\bibitem{Po}
G.S. Poltyrev, \emph{Coding in an asynchronous multiple-access channel},
  Problems Inform. Trans. \textbf{36} (1983), 12--21.

\bibitem{Sha2}
C.~E. Shannon, \emph{A mathematical theory of communication}, The Bell Sys.~
  Tech. Journal \textbf{27} (1948), 379--423.

\bibitem{TKW}
A.~Tchamkerten, A.~Khisti, and G.W. Wornell, \emph{Information theoretic
  perspectives on synchronization}, IEEE Intl. Sympo. on Info. Th. (ISIT),
  2006, pp.~371--375.

\bibitem{TT2}
A.~Tchamkerten and I.~E. Telatar, \emph{Variable length coding over an unknown
  channel}, IEEE Trans.~inform.~Th. \textbf{52} (2006), no.~5, 2126--2145.

\bibitem{Ve}
S.~Verd\'u, \emph{The capacity region of the symbol-asynchronous gaussian
  multiple-access channel}, IEEE Trans.~Inform.~Th. \textbf{35} (1989), no.~4,
  733--751.

\end{thebibliography}

\end{document}